\documentclass[10pt,oneside,onecolumn,a4paper,draftclsnofoot]{IEEEtran}
\usepackage{amsmath,epsfig,bbm}
\usepackage{amsmath}
\usepackage{gensymb}
\usepackage{graphicx}
\usepackage{color}
\usepackage[ruled]{algorithm2e}
\usepackage{multirow}

\def\comment#1{}
\newtheorem{theorem}{\bf Proposition}
\newtheorem{definition}{\bf Definition}

\newcommand {\eqdef}{\stackrel{\triangle}{=}}
\newcommand{\A}{{\bf A}}
\newcommand{\W}{{\bf W}}
\newcommand{\X}{{\bf X}}
\newcommand{\Y}{{\bf Y}}
\newcommand{\bS}{{\bf S}}
\newcommand{\HH}{{\bf H}}
\newcommand{\y}{{\bf y}}
\newcommand{\x}{{\bf x}}
\newcommand{\s}{{\bf s}}

\newcommand{\C}{{\bf C}}
\newcommand {\bLambda} {\boldsymbol{\Lambda}}
\newcommand {\bXi} {\boldsymbol{\Xi}}

\begin{document}

%\title{Scale-Invariant Reconstruction of Separated Sources}
\title{Performance Analysis of Source Image Estimators in Blind 
Source Separation}

\author{{\bf Zbyn\v{e}k Koldovsk\'{y} and Francesco Nesta}
\vspace{0.1in} \\
$^1$Faculty of Mechatronics, Informatics, and Interdisciplinary
Studies, Technical University of Liberec,\\ Studentsk\'a 2, 461 17
Liberec, Czech Republic. E-mail:
zbynek.koldovsky@tul.cz,\\ fax:+420-485-353112, tel:+420-485-353534\\
$^2$Conexant System, 1901 Main Street, Irvine, CA (USA). E-mail: 
francesco.nesta@conexant.com
}

\maketitle

\begin{abstract}
Blind methods often separate or identify signals or signal subspaces up to an 
unknown scaling factor. Sometimes it is necessary to cope with the scaling 
ambiguity, which can be done through reconstructing signals as they are 
received by sensors, because scales of the sensor responses (images) have known 
physical interpretations. 
In this paper, we analyze two approaches that are widely used for computing the 
sensor responses, especially, in 
Frequency-Domain Independent Component Analysis. One approach is the 
least-squares projection, while the other 
one assumes a regular mixing matrix and computes its inverse. Both estimators 
are invariant to the unknown scaling. Although frequently used, their 
differences were not studied yet.
A goal of this work is to fill this gap. The estimators are compared through a 
theoretical study, perturbation analysis and simulations.
We point to the fact that the estimators are equivalent when the 
separated signal subspaces are orthogonal, and vice versa. Two 
applications are shown, one of which demonstrates a case where the estimators 
yield substantially different results.

%, that the former approach 
%is less sensitive to identification errors provided that the covariance matrix 
%of the observed signals is precisely known. It is more practical in a sense 
%that the whole mixing matrix need not be identified for its use. The advantage 
%of the latter approach resides in the independence on the covariance matrix.
%Both approaches are equivalent under so-called orthogonal constraint.
%Applications from the area of noise reduction in speech and de-noising of signals from electrocardiogram are demonstrated.
\end{abstract}

{\keywords Beamforming, Blind Source Separation, Independent Component Analysis, Principal Component Analysis, Independent Vector Analysis}

\section{Introduction}
%Some blind array processing systems separate or identify signals or signal 
%subspaces up to unknown scaling factors. The so-called scaling ambiguity is 
%solved through reconstructing the separated signals as they are received by 
%sensors, because scales of the sensor responses (images) have known physical 
%interpretations. Moreover, some recent BSS methods consider the observed 
%signals directly as the sum of images of original sources, by which the 
%scaling 
%ambiguity is implicitly avoided \cite{cardoso2008,duong}. This motivates us 
%for 
%this work that is focused on two approaches that are widely used for computing 
%the signal images in BSS, especially, with Independent Component Analysis 
%(ICA).

The linear instantaneous complex-valued mixture model
\begin{equation}\label{model}
	{\bf x}={\bf H} {\bf s}\qquad\text{or}\qquad {\bf X}={\bf H} {\bf S}
\end{equation}
describes many situations where multichannel signals are observed, especially 
those considered in the field of array processing \cite{vantrees} and Blind 
Source Separation (BSS) \cite{BSS,cardoso}. 
The former equation is a vector-symbolic description of the model while the 
latter equation describes a batch of data.

The vector ${\bf x}=[x_1, \dots, 
x_d]^T$ represents $d$ observed signals on sensors, ${\bf s}=[s_1,\dots,s_r]^T$ 
represents original signals, and ${\bf H}$ is a $d\times r$ complex-valued 
mixing matrix representing the linear mixing system. Upper-case bold letters 
such as $\X$ and $\bS$ will denote matrices whose columns contain concrete 
samples of the respective signals; let the number of columns (samples) be $N$ 
where $N\gg d$.
We will focus on the regular case when the number of the observed signals $d$ 
is the same as that of the original signals $r$, but later in the article we 
will also address an undetermined case where $r>d$. From this point forward, 
let $\HH$ be a $d\times d$ full rank matrix. 

Consider a situation where only a subset of the original signals is of primary 
interest (e.g., only one particular source or the subspace spanned by some 
sources, so-called multidimensional source). 
Without a loss of generality, let $\s$ be divided into two components $[\s_1; \s_2]$ where the sub-vectors $\s_1$ and $\s_2$ have, respectively, length $m$ and $d-m$, $1\leq m < d$. The former component will be referred to as {\em target component}, and the latter as {\em interference}. Correspondingly, let $\HH$ be divided as $[\HH_1\, \HH_2]$ where the sub-matrices $\HH_1$ and $\HH_2$ have dimensions $d\times m$ and $d\times(d-m)$, respectively. 
Then (\ref{model}) can be written as
\begin{equation}\label{model2}
	{\bf x} = \HH_1\s_1 +  \HH_2\s_2.
\end{equation}
The terms $\HH_1\s_1$ and $\HH_2\s_2$ correspond to the contributions of $\s_1$ 
and $\s_2$, respectively, for the mixture ${\bf x}$, and will be denoted as 
${\bf s}^i\eqdef {\bf H}_i\,\s_i$, $i\in\{1,2\}$. If, for example, $\s_2$ is 
not active, then $\x=\s^1$, which is equal to the observations of $\s_1$ on the 
sensors, that is, the {\em sensor response} or {\em source image} of $\s_1$. 

In audio applications, $\s_1$ is often a scalar signal ($m=1$) representing a 
point source located in the room, and the model (\ref{model}) describes linear 
mixing in the frequency domain for a particular frequency bin 
\cite{duong,sawada,ukai,matsuoka}. For $m>1$, $\s_1$ can correspond to a 
subgroup of speakers \cite{schmulik}. In biomedical applications, 
$\s_1$ or $\s_2$ can consist of components related to a target activity such 
as muscular artifacts in electroencephalogram (EEG) \cite{eeg}, maternal or 
fetal electrocardiogram (ECG) \cite{ecglathauwer}, and so forth. 

The problem of retrieving ${\bf s}^i$ from $\x$ is often solved with the aid of 
methods for Blind Source Separation.
The objective of BSS is to separate the original signals based purely on their 
general properties (e.g., independence, sparsity or nonnegativity). In a 
general sense,
BSS involves Principal Component Analysis (PCA), Independent Component Analysis (ICA) \cite{BSS, ICA}, Independent Vector Analysis (IVA) \cite{iva1,iva2}, Nonnegative Matrix Factorization \cite{nmf}, etc. Some methods separate all of the one-dimensional components of $\s$ \cite{fastica,jade}, extract selected components only \cite{cichocki2004}, or separate multidimensional components; see, e.g., \cite{lahat2, isa, mica, lahat,tabcd}. 
The separation can also proceed in two steps where a steering vector/matrix 
(e.g., ${\bf H}_1$) is identified first, while the signals are separated in the 
second step using an array processor such as the minimum variance 
distortion-less (MVDR) beamformer \cite{capon}.

The blind separation or identification is often not unique. For example, 
the 
order and scaling factors of the separated components are random and cannot be 
determined without additional assumptions. 
Throughout this paper, we will always assume that the problem of the random 
order (known as the permutation problem) has already been resolved, so the 
subspaces in (\ref{model2}) are correctly identified; for practical methods 
solving the permutation problem, see, e.g., \cite{sawada,pairing,iva}. 

The reconstruction of the signal images is a way to cope with the scaling 
ambiguity \cite{matsuoka}. The advantage is that $\s^i$ can be retrieved 
without prior knowledge of the scale of $\s_i$ in (\ref{model}) or in 
(\ref{model2}). The scale of $\s^i$ has clear physical interpretation (e.g., 
voltage), so the retrieval is highly practical. For example, in the Frequency 
Domain ICA for audio source separation, the scaling ambiguity must 
be resolved within each frequency bin, which is important for reconstructing 
the spectra of the separated signals in the time domain \cite{sawada}. Some 
recent BSS 
methods aim to consider the observed signals directly as the sum of images of 
original sources, by which the scaling ambiguity is implicitly avoided and the 
number of unknown free parameters in the BSS model is decreased 
\cite{duong,cardoso2008,planck2013}. This motivates us for this 
study, because the way to reconstruct the signal images (either $\s^1$ or 
$\s^2$) is an important topic.

In this paper, we study two widely used methods to estimate the source images: 
One approach performs the least-squares projection of a separated source on 
the subspace spanned by $\X$. The other approach assumes that a blind estimate 
of a demixing transform is available and exploits its inverse to compute the 
sources' images. Both estimators are invariant to the unknown scaling of the 
separated sources. The goal of this study is to compare the estimators, which 
was not performed yet, and to provide a guidance which estimator is 
advantageous compared to the other from different aspects. We also show 
conditions under which the estimators are equivalent.

The following section introduces the estimators and points to their important 
properties and relations. 
Section III contains a perturbation analysis that studies cases where the 
estimated demixing transform contains ``small'' errors. Section IV studies 
properties of the least-squares estimator in underdetermined situations, that 
is, when there are more original signals than the observed ones. Section V 
presents results of simulations, and, finally, Section VI demonstrates two 
applications.

\section{Source Image Estimators}

Consider an exact demixing transform as a regular $d\times d$ matrix ${\bf W}$ 
defined as such that
\begin{equation}\label{separation}
{\bf W}{\bf H} = {\tt bdiag}(\bLambda_1,\bLambda_2),
\end{equation}
where $\bLambda_1$ and $\bLambda_2$ are arbitrary nonsingular matrices 
representing the random scaling factors of dimensions $m\times m$ and 
$(d-m)\times (d-m)$, 
respectively; ${\tt bdiag}(\cdot)$ denotes a block-diagonal matrix with the 
arguments on its block-diagonal. By applying $\W$ to $\x$, the outputs are
\begin{equation}\label{subspaces}
\y={\bf W}\x=\W\HH\s=\begin{pmatrix}\bLambda_1\s_1 \\ \bLambda_2\s_2 
\end{pmatrix}=\begin{pmatrix}\y_1 \\ \y_2 \end{pmatrix}.
\end{equation}
The components $\y_1=\bLambda_1\s_1$ and $\y_2=\bLambda_2\s_2$ are separated in 
the sense that each is a mixture only of $\s_1$ and $\s_2$, respectively. 

Let  $\W_1$ and $\W_2$ be sub-matrices of $\W$ such that $\W=[\W_1;\W_2]$, and 
$\W_1$ contains the first $m$ rows of $\W$, i.e., $\y_1=\W_1\x$. 
From (\ref{separation}) it follows that ${\bf W}$ is demixing if and only 
if\footnote{A more general definition is that ${\bf W}$ is demixing if and only 
if $\W_1\HH_2\s_2={\bf 0}$ and $\W_2\HH_1\s_1={\bf 0}$. However, we will assume 
that the 
mixing model (\ref{model}) is determined, so cases where 
$\W_1\HH_2\s_2={\bf 0}$ while $\W_1\HH_2\neq {\bf 0}$ and similar do not 
exist.}  
$\W_1\HH_2={\bf 0}$ and $\W_2\HH_1={\bf 0}$.

Throughout the paper we will occasionally use the following assumptions. 
Consider an estimated demixing matrix $\widehat\W$.
\begin{itemize}
	\item[\textsf{A1}($i$)] The assumption that $\widehat\W_i\HH_j={\bf 0}$ 
	where 
	$j\in\{1,2\}, 
	j\neq i$. Assuming an exact demixing transform thus corresponds to A1($1$) 
	simultaneously with A1($2$).%, will be referred to as the {\em 
	%orthogonality of $\W_i$}.
	\item[\textsf{A2}] The assumption of {\em uncorrelatedness of $\s_1$ and 
	$\s_2$} 
	means that ${\rm E}[\s_1\s_2^H]={\bf 0}$. %\footnote{Simultaneously, it is 
%	assumed here that the expectation is properly defined and that it exists.}.
	\item[\textsf{A3}] The assumption of {\em orthogonality of $\Y_1$ and 
	$\Y_2$}, that is,
	\begin{equation}\label{orthconst1}
	{\bf Y}_1{\bf Y}_2^H/N={\bf 0},
	\end{equation}
	in BSS also known as the orthogonal constraint \cite{cardoso1994}, 
	means that the sample-based estimate of ${\rm E}[\y_1\y_2^H]$ is exactly 
	equal to zero.
\end{itemize}
In the determined case $r=d$ and under \textsf{A1}(1) and \textsf{A1}(2), 
the condition \eqref{orthconst1} corresponds with 
\begin{equation}\label{orthconst2}
{\bf S}_1{\bf S}_2^H/N={\bf 0},
\end{equation}
but not generally so when $r>d$. The latter condition could be seen as a 
stronger alternative to \textsf{A2}.

For example, the orthogonal constraint \textsf{A3} is used by some ICA methods 
such as is Symmetric or Deflation FastICA \cite{fastica}. 
There are several reasons for this. First, \textsf{A2} is the necessary 
condition of independence of the original 
signals, so \textsf{A3} is a practical way to decrease the number of unknown 
parameters in ICA. Second, \textsf{A3} helps to achieve the global convergence 
(to find all independent components) and prevents algorithms from finding the 
same component twice. Finally, in the model \eqref{model2} with the \textsf{A3} 
constraint, $\widehat\W_i$ is already determined up to a 
scaling matrix when $\widehat\W_j$ is given, $j\neq i$, and vice versa.

\subsection{Reconstruction Using Inverse of Demixing Matrix}
The estimator to retrieve $\s^i$ described here will be abbreviated as INV. 
\begin{definition}[INV]
Let $\widehat\W= [\widehat\W_1 ; \widehat\W_2]$ denote an estimated demixing 
matrix by a BSS method, 
and let $\widehat\A$ be its inverse matrix, i.e., $\widehat\A \eqdef 
\widehat\W^{-1}$. Let  $\widehat\A = [\widehat\A_1 
\, \widehat\A_2]$ be divided in the same way as the system matrix $\HH$. Then, 
the INV
estimator is defined as
\begin{equation}\label{MDP}
	\widehat\s^i_{\rm INV}=\widehat\A_i\widehat\W_i\x\qquad\text{or}\qquad 
	\widehat\bS^i_{\rm INV}=\widehat\A_i\widehat\W_i\X,
\end{equation}
for $i\in\{1,2\}$.
\end{definition}

In particular, INV is popular in the frequency-domain ICA for audio source 
separation; see, e.g., \cite{matsuoka, ukai, francesco}. The following two 
propositions point to its important properties.

\begin{theorem}[consistency of INV]
Consider an exact demixing transform $\W$ satisfying \textsf{A1}($1$) and 
\textsf{A1}($2$), that 
is, satisfying (\ref{separation}); let $\A= [\A_1 \, \A_2]$ be its inverse 
matrix. For $i\in\{1,2\}$, it holds that
\begin{equation}
	\widehat\s^i_{\rm INV}=\A_i\W_i\x=\s^i.
\end{equation}
\end{theorem}
\begin{proof}
By (\ref{separation}) it holds that $\A_i=\HH_i\bLambda_i^{-1}$. Then,
\begin{align}
\A_i\W_i\x&=\A_i\W_i(\HH_1\s_1+\HH_2\s_2)\\ &=\A_i\W_i\HH_i\s_i \\
& = \HH_i\bLambda_i^{-1}\bLambda_i\s_i=\HH_i\s_i=\s^i.
\end{align}
\end{proof}

\begin{theorem}[scaling invariance of INV]
Let $\widehat\W$ be an estimated demixing transform and 
$\widehat\A=\widehat\W^{-1}$. The INV 
estimator is invariant to substitution $\widehat\W\leftarrow{\tt 
bdiag}(\bLambda_1,\bLambda_2)\widehat\W$, where $\bLambda_1$ and $\bLambda_2$ 
are 
arbitrary nonsingular matrices of dimensions $m\times m$ and 
$(d-m)\times (d-m)$, respectively.
\end{theorem}
\begin{proof}
The proof follows from the fact that $[\bLambda_1\widehat\W_1 ; 
\bLambda_2\widehat\W_2]^{-1}=[\widehat\A_1\bLambda_1^{-1} \, 
\widehat\A_2\bLambda_2^{-1}]$.
\end{proof}

One advantage is that the transform $\A_i\W_i$ is purely a function of ${\bf 
W}$ and does not explicitly depend on the signals or on their statistics, e.g., 
on the sample-based covariance matrix. This makes the approach suitable for 
real-time processing \cite{matsuoka2}.

On the other hand, $\A_i\W_i$ is a function of the whole $\W$ through the matrix inverse; it does not depend solely on $\W_i$, as one would expect when only $\s^i$ should be estimated. Formula (\ref{MDP}) can thus be used only if the whole demixing ${\bf W}$ is available. BSS methods extracting only selected components (e.g., one-unit FastICA \cite{fastica}) cannot be applied together with (\ref{MDP}). 
Next, it follows that potential errors in the estimate of $\W_2$ can have an 
adverse effect on the estimation of $\s^1$. This is analyzed in Section~III.

\subsection{Least-squares reconstruction}
Another approach to estimate $\s^i$ is to find an optimum projection of the 
separated components back to the observed signals $\x$ in order to find their 
contribution to them. A straightforward way is to use the quadratic criterion, 
that is, least squares, which gives two estimators that will be abbreviated by 
LS.

%The LS estimators operate with the estimated (part of) demixing matrix as well 
%as with signal samples. %To this end, let matrices of signal samples be 
%%denoted 
%by bold capital letters, e.g., let ${\bf X}$ denote the matrix whose columns 
%correspond to the samples of $\x$.

\begin{definition}[LS]
	Let $\widehat\W_i$ denote an estimated part of a demixing matrix, 
	$i\in\{1,2\}$, 
	$\y_i=\widehat\W_i\x$ and $\Y_i=\widehat\W_i\X$. The theoretical LS 
	estimator of $\s^i$ is 
	defined as
	\begin{align}
	\widehat\s^i_{\rm LS}&=\left(\arg\min_{\bf V}{\rm E}\left[\|\x-{\bf 
		V}\y_i\|^2\right]\right)\x\label{LSsymb}\\
	&=\C\widehat\W_i^H(\widehat\W_i\C\widehat\W_i^H)^{-1}\widehat\W_i\x,
	\end{align}
	where $\C={\rm E}[\x\x^H]$. The practical LS estimator of $\bS^i$ is 
	defined as
	\begin{align}
	\widehat\bS^i_{\rm LS}&=\left(\arg\min_{\bf V}\|\X-{\bf 
	V}\Y_i\|_F^2\right)\X\label{LS}\\
	&=\widehat\C\widehat\W_i^H(\widehat\W_i\widehat\C\widehat\W_i^H)^{-1}
	\widehat\W_i\X,
	\end{align}
	where $\widehat\C=\X\X^H/N$, and $\|\cdot\|_F$ denotes the Frobenius norm. 
\end{definition}

\begin{theorem}[scaling invariance of LS]
	The estimators (\ref{LSsymb}) and (\ref{LS}) are invariant to a scaling 
	transform $\widehat\W_i\leftarrow\bLambda_i\widehat\W_i$ where $\bLambda_i$ 
	is a 
	nonsingular 
	square matrix. 
\end{theorem}

The proof of Proposition~3 is straightforward. It is worth pointing out that the LS estimators
are purely functions of $\widehat\W_i$ as compared to INV. Also, they involve 
the 
covariance matrix or its sample-based estimate. However, their consistency is 
not guaranteed under \textsf{A1}($i$) even if the assumption holds for both 
$i=1,2$ as assumed in Proposition~1. In fact, additional assumptions are needed 
as is 
shown by the following proposition.

\begin{theorem}
Let $\W_i$ be a part of an exact demixing matrix, so \textsf{A1}($i$) holds. 
Let $\W_i\HH_i=\bLambda_i$ be nonsingular. 
Then, under \textsf{A2} it holds that
\begin{equation}\label{consistencyLS1}
	\widehat\s^i_{\rm LS}=\C\W_i^H(\W_i\C\W_i^H)^{-1}\W_i\x=\s^i,
\end{equation}
and under \textsf{A3}  it holds that
\begin{equation}\label{consistencyLS2}
	\widehat\bS^i_{\rm 
	LS}=\widehat\C\W_i^H(\W_i\widehat\C\W_i^H)^{-1}\W_i\X=\bS^i.
\end{equation}
\end{theorem}
\begin{proof}
The proof will be given for (\ref{consistencyLS1}) while the one for (\ref{consistencyLS2}) is analogous.

According to (\ref{model}) it holds that $\C={\rm E}[\x\x^H]=\HH\,\C_{\s}\HH^H$ 
where $\C_\s={\rm E}[\s\s^H]$. Under \textsf{A2} it follows that $\C_{\s}$ has 
the  
block-diagonal structure $\C_{\s}={\tt bdiag}(\C_{\s_1},\C_{\s_2})$ where 
$\C_{\s_1}\eqdef{\rm E}[\s_1\s_1^H]$ and $\C_{\s_2}\eqdef{\rm E}[\s_2\s_2^H]$ 
are regular (because $\C_{\s}$ is assumed to be regular). Without a loss of 
generality, let $i=1$. Since $\W_1\HH=(\bLambda_1\,{\bf 0})$,
\begin{align}
&\widehat\s^i_{\rm LS}=\C\W_1^H(\W_1\C\W_1^H)^{-1}\W_1\x \\&= 
\HH\C_{\s}\HH^H\W_1^H(\W_1\HH\C_{\s}\HH^H\W_1^H)^{-1}\W_1\HH\s \\
& = \HH_1\C_{\s_1}\bLambda_1^H(\bLambda_1\C_{\s_1}\bLambda_1^H)^{-1}\bLambda_1\s_1\\
& = \HH_1\s_1 = \s^1.
\end{align}
\end{proof}

It is worth pointing out that LS involves a matrix inverse, 
namely, of $\W_i\C\W_i^H$ or of $\W_i\widehat\C\W_i^H$. This matrix (actually, 
the (sample) covariance of ($\Y_i$) $\y_i$) has a lower dimension than $\W$ and 
is more likely well conditioned so that the computation of its inverse is 
numerically stable.

\subsection{On the equivalence between INV and LS under the orthogonal 
constraint}
Without a loss of generality, assume that $\widehat{\bf W}_1$ is given.
Under the assumption \textsf{A3}, $\widehat\W_2$ is already determined up to 
a scaling matrix through \eqref{orthconst1}, so the whole $\widehat{\W}$ is 
actually available, and the 
INV estimator (\ref{MDP}) can be applied. The goal here is to verify that, in 
this case, INV coincides with LS.

%Let $\bLambda_1=\W_1\HH_1$ be the unknown scaling of $\bS_1$, and 
Let ${\bf B}$ denote the unknown lower part of the entire demixing 
$\widehat{\bf 
W}=[\widehat\W_1; {\bf B}]$. 
Then,
\begin{equation}
	\widehat{\bf W}{\bf X}=
	\begin{pmatrix}
	\widehat\W_1\X \\ {\bf B}\X
	\end{pmatrix}=
%	\begin{pmatrix}
%	\bLambda_1\bS_1 \\ 
%	{\bf B}\HH_1 \bS_1 + {\bf B}{\bf H}_{2}{\bf S}_{2}
%	\end{pmatrix}=
	\begin{pmatrix}
	\Y_1 \\ \Y_2
	\end{pmatrix}.
\end{equation}
The condition \eqref{orthconst1} %, under A3, it holds that $\W$ is demixing if 
%and only if $\Y_1\Y_2^H={\bf 0}$. The latter condition requires that
requires that
\begin{equation}\label{cond2}
	{\bf B}\widehat\C\widehat\W_1^H={\bf 0},
\end{equation}
which means that the rows of ${\bf B}$ are orthogonal to the columns of 
$\widehat\C\widehat\W_1^H$. It can be verified that any ${\bf B}$ of the form
\begin{equation}
	{\bf B} = {\bf Q}({\bf 
	I}-\widehat\C\widehat\W_1^H(\widehat\W_1\widehat\C\widehat\C\widehat\W_1^H)^{-1}
	\widehat\W_1\widehat\C),
\end{equation}
where ${\bf Q}$ can be an arbitrary $(d-m)\times m$ full-row-rank matrix such that ${\bf B}$ has full row-rank, meets the condition (\ref{cond2}).

Now, to apply (\ref{MDP}), $\widehat\A_1$ must be computed, which consists of 
first $m$ 
columns of $\widehat\A=\widehat{\bf W}^{-1}$, so it satisfies
\begin{align}
\widehat\W_1\widehat\A_1&={\bf I},\label{cond3}\\
{\bf B}\A_1&={\bf Q}({\bf 
I}-\widehat\C\widehat\W_1^H(\widehat\W_1\widehat\C\widehat\C\widehat\W_1^H)^{-1}
\widehat\W_1\widehat\C)\widehat\A_1={\bf  0}.
\end{align}
The latter equation is satisfied whenever $\widehat\A_1=\widehat{\bf 
C}\widehat\W_1^H{\bf R}$ where ${\bf R}$ is an $m\times m$ matrix. To satisfy 
also (\ref{cond3}), ${\bf R}=(\widehat\W_1\widehat\C\widehat\W_1^H)^{-1}$. 
Finally, 
\begin{equation}
\bS^1_{\rm INV}=\widehat\A_1\widehat\W_1\X=\widehat{\bf 
C}\widehat\W_1^H(\widehat\W_1\widehat\C\widehat\W_1^H)^{-1}\widehat\W_1\X,
\end{equation}
which allows us to conclude this section by the following proposition.

\begin{theorem}
Let $\widehat\W_i$ be a part of an estimated demixing matrix, $i\in\{1,2\}$, 
and let \textsf{A3} be assumed. 
Then, $$\widehat\bS^i_{\rm INV}=\widehat\bS^i_{\rm LS}.$$
\end{theorem}

\section{Perturbation Analysis}
%\begin{table*}
%\caption{\label{table1} Comparison of expressions (\ref{MDPan2}) and 
%(\ref{proposedan3})}
%\centering
%\begin{tabular}{l|c|c|c}
%~ & $m=1$ & $m=d/2$ & $m=d-1$\\
%\hline
%INV: Expression (\ref{MDPan2}) & $\lambda_1^2+(d-1)\lambda_2^2$ & 
%$\frac{d^2}{4}(\lambda_1^2+\lambda_2^2)$ & 
%$(d-1)^2\lambda_1^2+(d-1)\lambda_2^2$ \\
%\hline
%LS: Expression (\ref{proposedan3}) & $2(d-1)\lambda_1^2$ & 
%$\frac{d^2}{2}\lambda_1^2$ & $2(d-1)\lambda_1^2$ 
%\end{tabular}
%\end{table*}

Throughout this section, let ${\bf W}$ denote the exact demixing transform, 
that is ${\bf W}={\bf H}^{-1}$. We present an analysis of the sensor response 
estimators (\ref{MDP}) and (\ref{LS}) when ${\bf W}$ is known 
up to a small deviation\footnote{The analysis of (\ref{LSsymb}) follows from 
that of (\ref{LS}) when $\Delta\C={\bf 0}$.}. Let ${\bf 
V}={\bf W}+\bXi$ be the available 
observation of ${\bf W}$ where $\bXi$ is a ``small'' matrix;% in the sense that 
%its elements have stochastic asymptotic order $\mathcal{O}_p(N^{-1/2})$; 
${\bf V}_1$ will denote the sub-matrix of ${\bf V}$ containing the first $m$ 
rows; similarly $\bXi=[\bXi_1;\bXi_2]$; ${\bf A}={\bf V}^{-1}$ and ${\bf A}_1$ 
contains the first $m$ columns of ${\bf A}$; let also 
$\Delta\C=\widehat\C-\C$ be ``small'' and of the same asymptotic order as 
$\bXi$ (typically, $\Delta\C=\mathcal{O}_p(N^{-1/2})$ where 
$\mathcal{O}_p(\cdot)$ is the stochastic order symbol).

Now, consider the transform matrices ${\bf T}_1\eqdef\A_1{\bf V}_1$ %, ${\bf 
%T}_2\eqdef{\bf C}{\bf V}_1^H({\bf V}_1\C{\bf V}_1^H)^{-1}{\bf V}_1$, 
and ${\bf 
T}_2\eqdef\widehat{\bf C}{\bf V}_1^H({\bf V}_1\widehat\C{\bf V}_1^H)^{-1}{\bf 
V}_1$ estimating $\bS^1$ from $\X$, respectively. 
%From here, the estimators (\ref{MDP}) and (\ref{proposed}) will be referred to 
%as INV (an estimator using system matrix INVerse) and SRR (Sensor Response 
%Reconstruction), respectively. 
The analysis resides in the computation of their squared distances (the 
Frobenius norm) from the ideal transform, that is, from ${\bf H}_1{\bf W}_1$. 
Using first-order expansions and neglecting higher-order terms, it is derived 
in Appendix that the following approximations hold.
\begin{align}
\left\|\HH_1\W_1-{\bf T}_1\right\|_F^2&\approx\left\|{\bf 
H}\,\bXi\,\HH_1\W_1-\HH_1\bXi_1\right\|_F^2, \label{MDPan1}\\
%\end{align}
%\begin{align}
\left\|\HH_1\W_1-{\bf T}_2\right\|_F^2&\approx
\Bigl\|\HH_1(\bXi_1\C\W_1^H+\W_1\C\,\bXi_1^H)\C_{\s_1}^{-1}\W_1\nonumber\\
(\HH_1\W_1-{\bf 
I})\Delta\C&\W_1^H\C_{\s_1}^{-1}\W_1-\HH_1\bXi_1-\C\,\bXi_1^H
\C_{\s_1}^{-1}\W_1\Bigr\|_F^2\label{proposedan1}.
\end{align}
%where $\Delta\C=\widehat\C-\C$; the expression for $\left\|\HH_1\W_1-{\bf 
%T}_2\right\|_F^2$ follows from (\ref{proposedan1}) by setting $\Delta\C={\bf 
%0}$.

To provide a deeper insight, we will analyze a particular case where 
$\HH=\W={\bf I}$, $\C_{\s_1}=\sigma_1^2{\bf I}$, and 
$\C_{\s_2}=\sigma_2^2{\bf I}$. Let the elements of $\bXi$ all be independent 
random variables with zero mean such that the variance of each element of 
$\bXi_i$ is equal to $\lambda_i^2$. For further simplification, let the 
elements of $\Delta\C_{ij}$, which denotes the $ij$th block of $\Delta\C$, 
$i,j=1,2$, be also independent random variables
%\footnote{Typically, the elements of $\Delta\C$ and $\bXi$ are dependent as 
%these matrices contain estimation errors due to the same data $\X$.} 
whose variance is equal to 
$\sigma_i\sigma_jC$.
Then, the expectation values of the right-hand sides of (\ref{MDPan1}) and 
(\ref{proposedan1}), 
respectively, are equal to
\begin{align}\label{MDPan2}
{\rm E}\left[\left\|\HH_1\W_1-{\bf T}_1\right\|_F^2\right]&\approx 
(\lambda_1^2+\lambda_2^2)m(d-m),\\
\label{proposedan2}
{\rm E}\left[\left\|\HH_1\W_1-{\bf T}_2\right\|_F^2\right]&\approx
\left[\left(1+\frac{\sigma_2^2}{\sigma_1^2}\right)\lambda_1^2+
\frac{\sigma_2}{\sigma_1}C\right]m(d-m).
\end{align}

Comparing (\ref{MDPan2}) and (\ref{proposedan2}) shows the pros and cons of the 
estimators. The latter depends on $\sigma_2^2/\sigma_1^2$, which reflects the 
ratio between the power of $\s_1$ and that of $\s_2$. The expression 
(\ref{MDPan2}) does not depend on this ratio explicitly\footnote{Typically, 
there is an implicit dependency of the estimation error $\bXi$ on 
$\sigma_2^2/\sigma_1^2$. Therefore, $\lambda_1^2$ as well as 
$\lambda_2^2$ are influenced by $\sigma_1^2$ and $\sigma_2^2$.}. For 
simplicity, let us assume that $\sigma_2^2/\sigma_1^2=1$.

Next, (\ref{proposedan2}) depends on $C$ while (\ref{MDPan2}) is independent of 
it. Since $C$ captures the covariance estimation error in $\Delta\C$, it 
typically decreases with the length of data $N$. Usually, $C$ has asymptotic 
order $\mathcal{O}(N^{-1/2})$; see, e.g., Appendix A.B in \cite{analysisTSP}.
For the sake of the analysis, we will assume that $C=0$. Then, 
(\ref{proposedan2}) changes to
\begin{equation}\label{proposedan3}
{\rm E}\left[\left\|\HH_1\W_1-{\bf T}_2\right\|_F^2\right]\approx
2\lambda_1^2m(d-m).
\end{equation}

The expressions (\ref{MDPan2}) and (\ref{proposedan3}) 
point to the main difference of the estimators when $\HH=\W={\bf I}$: While the 
performance of INV depends on $\lambda_1^2$ and $\lambda_2^2$, that of LS 
depends purely on $\lambda_1^1$. In the special case $\lambda_1^1=\lambda_2^1$, 
the performances coincide.

To verify the theoretical expectations, we conducted a simple simulation where 
$d=20$, $m=5$, $\HH=\W={\bf I}$, $\sigma_1^2=\sigma_2^2=1$, $C=0$, $N=10^5$, 
$\lambda_2^2=10^{-4}$. $\bXi$ and $\bS$ were drawn from complex Gaussian 
distribution with zero mean and unit variance. The average 
squared errors of INV and LS averaged over $100$ trials for each value of 
$\lambda_1$ are compared, respectively, with the expressions (\ref{MDPan2}) 
and (\ref{proposedan3}) in Fig~\ref{fig1}.

The theoretical error of INV is in a good agreement with the experimental one 
for all values of $\lambda_1$. The same holds for that of LS until 
$\lambda_1\geq 0.003$. For smaller values of $\lambda_1$, the experimental 
error of LS is limited while the theoretical one is decreasing. This is caused 
by the influence of $\Delta\C$, which is fully neglected in (\ref{proposedan3}) 
and modeled through $C$ in (\ref{proposedan2}). The experimental error of LS 
evaluated for $N=10^6$ confirms that (\ref{proposedan2}) is a more accurate 
theoretical error of LS than (\ref{proposedan3}).

In this example, INV outperforms LS when $\lambda_1>\lambda_2$, and vice versa. 
However, these results are valid only in the 
special case when $\HH=\W={\bf I}$. Simulations in Section~V consider general 
mixing matrices, thereby compare the estimators in more realistic situations.

%Table~\ref{table1} compares the expressions (\ref{MDPan2}) and 
%(\ref{proposedan3}) in three different cases. For $m=1$, which means that the 
%target component $\s_1$ is a scalar signal, INV is twice as accurate as LS 
%when 
%$\lambda_1^2\approx \lambda_2^2$. However, when $\lambda_2^2\gg \lambda_1^2$, 
%which means that the target component is estimated with higher accuracy than 
%the interference, the accuracy of INV is becoming deteriorated while that of 
%LS is independent of $\lambda_2^2$.

%For $m=d/2$ (assuming $d$ is even), the estimators behave similarly when 
%$\lambda_1^2\approx \lambda_2^2$. They thus perform approximately the same 
%until $\lambda_1^2\geq \lambda_2^2$. LS profits from its independence of 
%$\lambda_2^2$ compared to INV when $\lambda_2^2>\lambda_1^2$.
%
%The last case, when $m=d-1$, corresponds to a situation where the observed 
%data are contaminated by a scalar (low-dimensional) signal. Here, INV is more 
%sensitive to $\lambda_1^2$ and its accuracy is becoming worse with the growing 
%dimension $d$. 
%
%Fig.~\ref{fig1} shows the ratio of the expressions (\ref{MDPan2}) and 
%(\ref{proposedan3}), i.e., the relative mean-squared error, as the function of 
%$\lambda_1^2/\lambda_2^2$. It reveals that INV is suffering from small 
%values of the ratio $\lambda_1^2/\lambda_2^2$ while the accuracy of LS is 
%independent of $\lambda_2^2$. On the other hand, INV may outperform LS, 
%especially, when $\lambda_1^2\gg\lambda_2^2$ and $m\leq d/2$. In practice, 
%also 
%the influence of $\Delta\C$ contributing to the mean-squared error of LS must 
%be 
%taken into account.

\begin{figure}
\centering
\includegraphics[width=0.95\linewidth]{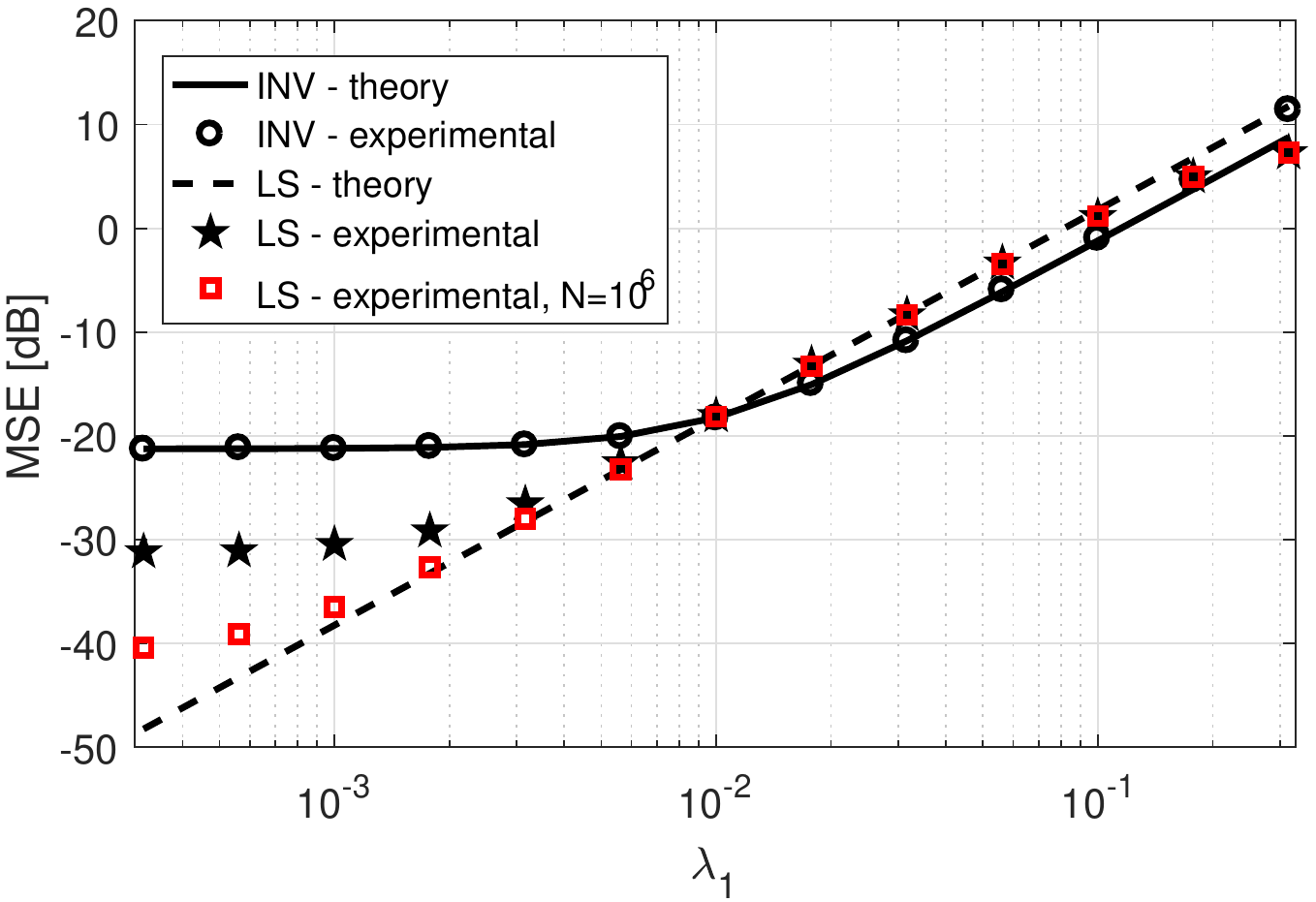}
\caption{\label{fig1} Comparison of theoretical and experimental mean square 
error of INV and LS when $d=20$, $m=5$, $\HH=\W={\bf I}$, 
$\sigma_1^2=\sigma_2^2=1$, $C=0$, $N=10^5$, $\lambda_2^2=10^{-4}$.}
\end{figure}

\section{Noise Extraction from Underdetermined 
Mixtures}\label{section:underdetermined}

\subsection{Mixture Model}
Now we focus on a more realistic scenario that appears in most array processing 
problems.  %a linear mixture of $m$ signals of interest is observed through $d$ 
%sensors and each observed signal is disturbed by noise. The 
Let the mixture be described as
\begin{equation}\label{umodel}
 {\bf x}={\bf H}_1 {\bf s}_1 + \s_2,
\end{equation}
where $\HH_1$ is a $d\times m$ matrix having full column rank, $\s_1$ is an 
$m\times 1$ vector of target components, and $\s_2$ is a $d\times 1$ vector of 
noise signals. %These models are often considered in frequency-domain audio 
%noise reduction systems \cite{STFT}. 
Note that, in this model, $\s_2$ is simultaneously equal to $\s^2$.

The mixture model corresponds with (\ref{model}), but $\HH$ is equal to 
$[\HH_1\,{\bf I}_{d\times d}]$ and has dimensions $d\times(m+d)$, which makes 
the problem underdetermined ($r=m+d$). 

%When the noise signals are assumed to be 
%uncorrelated with $\s_1$, that is, when the assumption \textsf{A2} holds,
%the covariance of $\x$ reads 
%\begin{equation}\label{covx}
%\C=\HH_1\C_{\s_1}\HH_1^H + \C_{\s_2}.
%\end{equation}
%Note that even under the stronger assumption \textsf{A3}, the analogy {\em 
%%%does 
%not} hold for the sample-based 
%covariances, i.e.,
%\begin{equation}\label{covX}
%\widehat\C\neq\HH_1\widehat\C_{\s_1}\HH_1^H + \widehat\C_{\s_2}.
%\end{equation}

In general, a linear transform that separates $\s_1$ or $\s_2$ from $\x$ does 
not exist, unless 
(\ref{umodel}) is implicitly regular (e.g., when $\C_{\s_2}$ has rank 
$d-m$)\footnote{For example, model (\ref{umodel}) is often studied under the 
assumption that at most $d$ signals out of ${\bf s}_1$ and  $\s_2$ are active 
at a given time instant; see, e.g. \cite{tifrom,undetermined1}.}. From now on, 
we focus on the difficult case where, generally speaking, neither $\s_1$ nor 
$\s_2$ can be separated. 

\subsection{Target Signal Cancelation and Noise Extraction}
Since the separation of $\s_1$ is not possible, multichannel noise reduction 
systems follow an inverse approach: the target components $\s_1$ are first 
linearly canceled from the mixture in order to estimate a reference of the 
noise components $\s_2$. Second, a linear transform or adaptive filtering is 
used to subtract the 
noise from the mixture as much as possible; see, e.g., \cite{GSC, cancelation, 
hoshuyama, 
takahashi, gannot, postfilters}. %However, the scaling uncertainty makes the 
%final adaptive filtering difficult, since the scaling needs to be estimated 
%from the data.

Specifically, the cancelation of the target component is achieved through a 
matrix $\W$ such that
\begin{equation}\label{blocking}
	\W\HH_1={\bf 0}.
\end{equation}
Since $\HH_1$ has rank $m$, the maximum possible rank of $\W$ is $d-m$, which 
points to the fundamental limitation: The maximum dimension of the subspace 
spanned by the extracted noise signals $\W\x=\W\s_2$ is $d-m$.

Assume for now that any $(d-m)\times d$ matrix $\W$ having full row-rank has 
been identified (e.g., using BSS). To estimate $\s_2$, LS can be used (INV 
cannot be applied in the underdetermined case), so
\begin{equation}\label{est2}
	\widehat\s^2_{\rm LS}=\C{\bf W}^H({\bf W}{\bf C}{\bf W}^H)^{-1}{\bf W}\x,
\end{equation}
or
\begin{equation}\label{est3}
\widehat\bS^2_{\rm LS}=\widehat\C{\bf W}^H({\bf W}\widehat{\bf C}{\bf 
W}^H)^{-1}{\bf 
W}\X.
\end{equation}

\begin{theorem}\label{propMSE}
Let $\W$ be a $(d-m)\times d$ transform matrix having rank $d-m$ and satisfying 
(\ref{blocking}), and let ${\bf Q}$ denote a $d\times (d-m)$ matrix. Under 
\textsf{A2}, 
$\s^2_{\rm LS}$ is 
a minimizer of
\begin{equation}\label{min1}
\min_{\widehat\s={\bf Q}\W\x} {\rm E}\left[ \|\s_2-\widehat\s\|^2\right].
\end{equation}
Assuming (\ref{orthconst2}), $\bS^2_{\rm LS}$ is a minimizer of
\end{theorem}
\begin{equation}\label{min2}
\min_{\widehat\bS={\bf Q}\W\X} \|\bS_2-\widehat\bS\|_F^2.
\end{equation}

\begin{proof}
Under \textsf{A2} it holds that
\[
\min_{\widehat\s={\bf Q}\W\x} {\rm E}\left[ \|\s_2-\widehat\s\|^2\right]=
\min_{\widehat\s={\bf Q}\W\x} {\rm E}\left[ \|\x-\widehat\s\|^2\right].
\]
When (\ref{orthconst2}) holds, than
\[
\min_{\widehat\bS={\bf Q}\W\x} \|\bS_2-\widehat\bS\|_F^2=
\min_{\widehat\bS={\bf Q}\W\x} \|\X-\widehat\bS\|_F^2.
\]
The statements of the proposition follow, respectively, by the definitions 
\eqref{LSsymb} and \eqref{LS}.
\end{proof}

The latter proposition points to important limitations of LS 
that should be taken into account in the 
underdetermined scenario. First, LS 
estimates a $d$-dimensional signal only from a $(d-m)$-dimensional signal 
subspace. Second, \eqref{est2} is optimal under \textsf{A2} in the 
least-squares sense of \eqref{min1}. Third, \eqref{est3} is optimal in the 
sense of \eqref{min2} only when (\ref{orthconst2}) is valid, which is much 
stronger assumption than \textsf{A2}.

\section{Simulations}
This section is devoted to extensive Monte Carlo simulations where the signals and system parameters are randomly generated. Real and complex parts of random numbers are always generated independently according to the Gaussian law with zero mean and unit variance. Each trial of a simulation consists of the following steps.
\begin{enumerate}
  \item The dimension parameters $d$ and $m$ are chosen.
	\item %$N=10^4$ 
	$N$ samples of the original components $\s_1$ and $\s_2$ are randomly 
	generated according to the Gaussian law. %Then, they are multiplied by 
	%square matrices ${\bf Q}_1$ and ${\bf Q}_2$ of corresponding dimensions, 
	%respectively. These matrices introduce inter-channel dependencies inside 
	%the multidimensional components $\s_1$ and $\s_2$. 
	\item The mixing matrix $\HH$ is generated, $\W=\HH^{-1}$, $\X=\HH\bS$, and 
	$\widehat\C=\X\X^H/N$.
	\item The estimation of $\W$ is simulated by adding random perturbations to 
	its blocks, that is, $\widehat\W_1=\W_1+\bXi_1$ and 
	$\widehat\W_2=\W_2+\bXi_2$, where the elements of $\bXi_1$ and $\bXi_2$ 
	have, respectively, variances $\lambda_1^2$ and $\lambda_2^2$; 
	$\widehat\W=[\widehat\W_1\,\widehat\W_2]$\footnote{Note that $\widehat\W_1$ 
	and $\widehat\W_2$ are not constrained to satisfy \textsf{A3}. Otherwise, 
	the comparison of INV and LS would not give a sense due to Proposition~5.}. 
	Then, $\widehat\W_1$ and $\widehat\W_2$ are multiplied by random regular 
	scaling matrices of corresponding dimensions.
	\item The accuracy of the INV and LS estimates of $\bS^1$ using 
	$\widehat\W$ is evaluated through the normalized mean-squared error defined 
	as
	\begin{equation}\label{NMSE}
		{\rm NMSE}_j=\frac{\|\HH_1\W_1-{\bf T}_j\|_F^2}{\|\HH_1\W_1\|_F^2},
	\end{equation}
	$j=1,2$.
\end{enumerate}

The following subsection reports results of simulations assuming the determined model. The next subsection considers the underdetermined model (\ref{umodel}).

\subsection{Determined model}

\subsubsection{Influence of the Estimation Errors in $\widehat\W$}
\begin{figure}
\centering
\includegraphics[width=\linewidth]{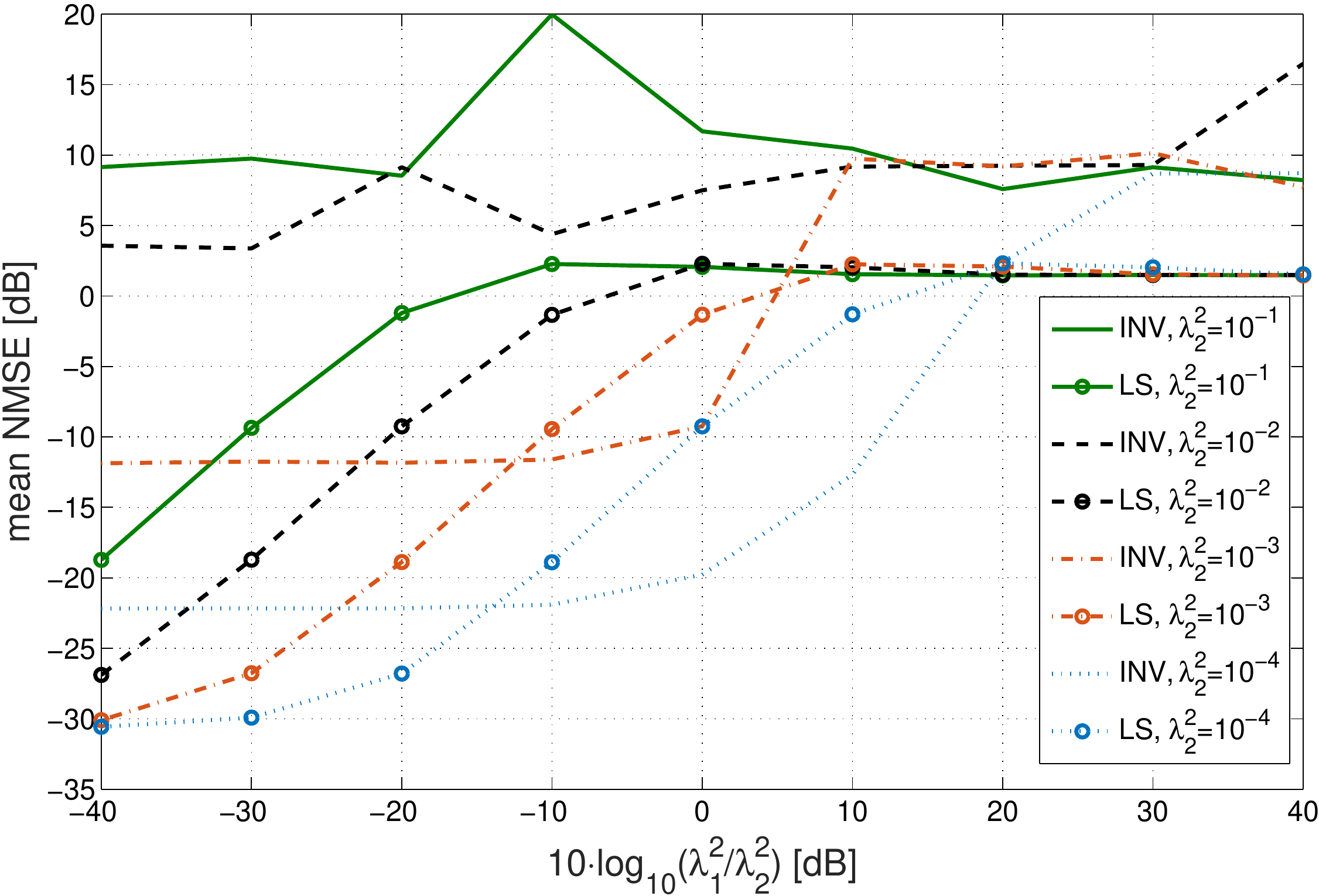}
\caption{\label{art1} NMSE averaged over $10^5$ trials where $d=5$, $m=2$,  
$N=10^4$, $\lambda_2^2$ is fixed, and $\lambda_1^2$ is varied.}
\end{figure}

The experiment is done with $d=5$, $m=2$, and $N=10^4$; $\lambda_2^2$ is equal 
to one of four constants ($10^{-1}$, $10^{-2}$, $10^{-3}$, and $10^{-4}$) while 
$\lambda_1^2$ is varied. Each simulation is repeated in $10^5$ trials. The 
average NMSE achieved by INV and LS are shown in Fig.~\ref{art1}. 

The results of INV are highly influenced by $\lambda_2^2$ that controls the 
perturbation of $\widehat\W_2$. For example, for $\lambda_2^2=10^{-1}$ and 
$\lambda_2^2=10^{-2}$, INV fails in the sense that the achieved NMSE is above 
$0$~dB. This happens even if $\lambda_1^2$, which controls the perturbation of 
$\widehat\W_1$, is relatively ``small''. For $\lambda_2^2=10^{-3}$ and 
$\lambda_2^2=10^{-4}$, the NMSE of INV decreases  with decreasing 
$\lambda_1^2$. However, the NMSE is lower bounded (does not improve as 
$\lambda_1^2\rightarrow 0$). All these results point to the dependency of INV 
on $\widehat\W_2$.

The NMSE of LS depends purely on $\lambda_1^2$. It is always improved with the 
decreasing value of $\lambda_1^2$ (it is only limited by the length of data 
which influences the accuracy of the sample covariance matrix $\widehat\C$). In 
this experiment, LS is outperformed by INV only in a few cases, namely, when 
$\lambda_2^2=10^{-4}$ and $\lambda_1^2/\lambda_2^2$ is higher than $-14$~dB. 
INV thus appears to be beneficial compared to LS in situations where the whole 
$\widehat\W$ is a sufficiently accurate estimate of $\W$.

\subsubsection{Varying Dimension}
\begin{figure}
\centering
\includegraphics[width=\linewidth]{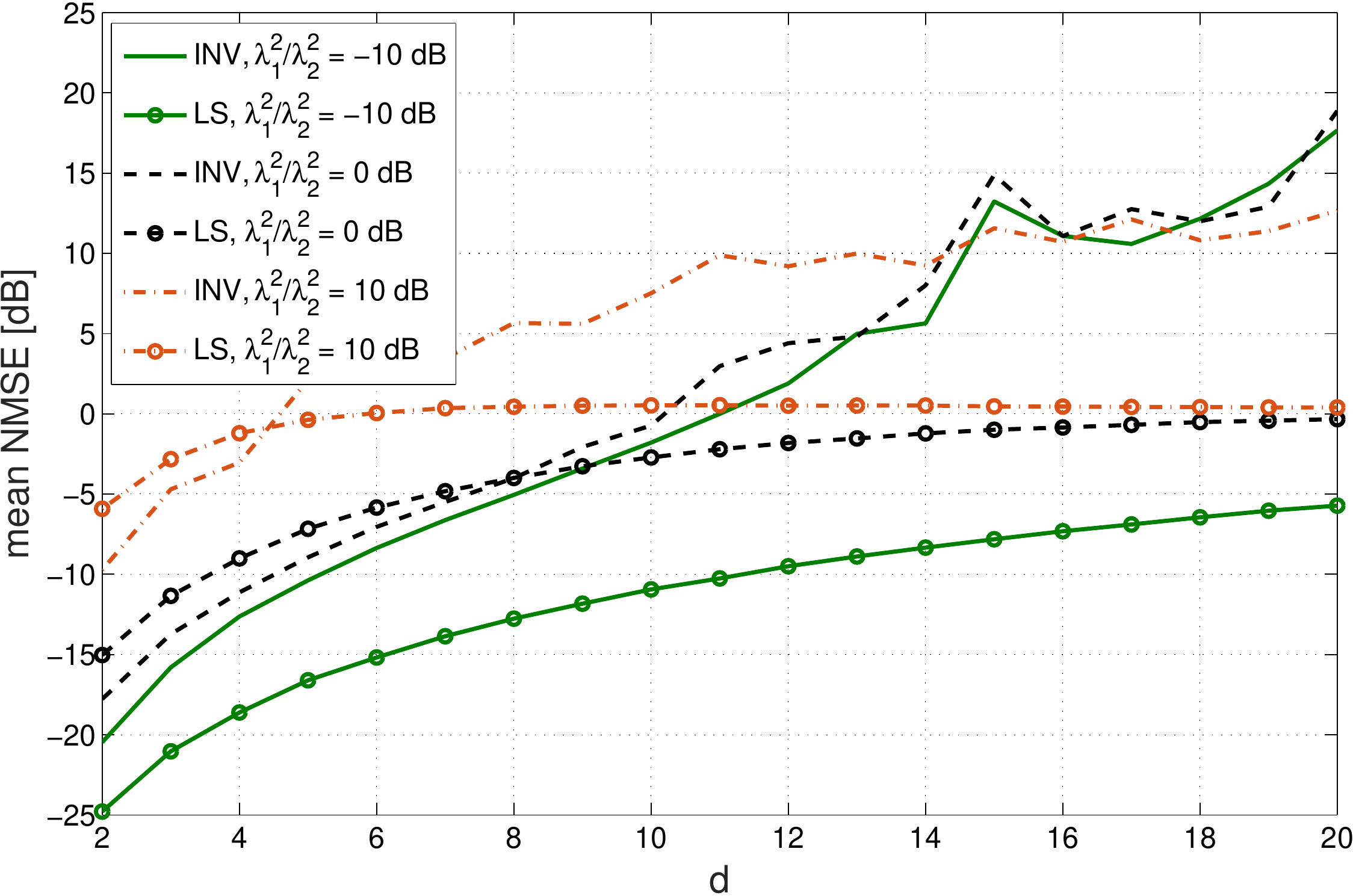}
\caption{\label{art3} NMSE averaged over $10^5$ trials as a function of 
$d=2,\dots,20$; here $m=1$, $\lambda_2^2=10^{-3}$, and $N=10^4$.}
\end{figure}

In the situation here, the target component $\s_1$ has dimension one, i.e., 
$m=1$, while the dimension of the mixture $d$ is changed from $2$ through $20$; 
$N=10^4$. The variances $\lambda_1^2$ and $\lambda_2^2$ are fixed, namely, 
$\lambda_2^2=10^{-3}$ and $\lambda_1^2$ is chosen such that 
$10\log_{10}\lambda_1^2/\lambda_2^2$ corresponds, respectively, to $-10$, $0$, 
and 
$10$~dB. The NMSE averaged over $10^5$ trials is shown in Fig.~\ref{art3}.

The NMSE values of both methods are increasing with growing $d$. In the INV case, the NMSE grows smoothly until it reaches a certain threshold value of $d$. The experiments show that this threshold depends on $\lambda_1^2$ and $\lambda_2^2$. Above this threshold, the NMSE of INV abruptly grows. It points to a higher sensitivity of INV to the estimation errors in $\widehat\W$ when the dimension of data is ``high''. 

LS yields smooth and monotonic behavior of NMSE for every $d$. It is 
outperformed by INV if both $\lambda_1^2$  and $\lambda_2^2$ as well as the 
data dimension $d$ are sufficiently small. 

\subsubsection{Target Component Dimension}
\begin{figure}
\centering
\includegraphics[width=\linewidth]{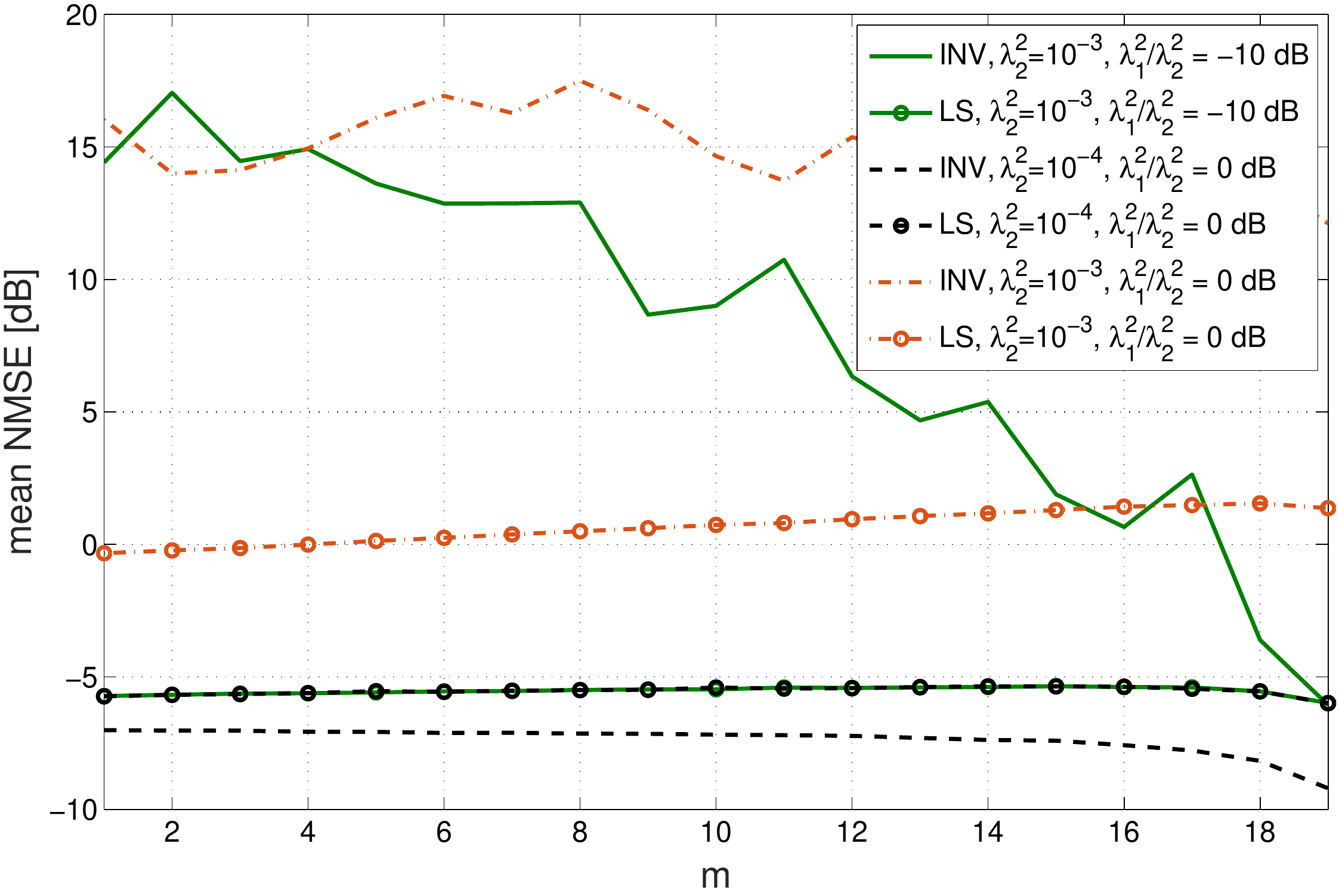}
\caption{\label{art2} NMSE averaged over $10^5$ trials where $d=20$, $N=10^4$, 
and $m=1,\dots,19$.}
\end{figure}

The dimension of the mixture $d$ is now put equal to $20$, while the dimension of the target component $m$ is varied from $1$ through $d-1$. Results for three different choices of $\lambda_1^2$ and $\lambda_2^2$ are shown in Fig.~\ref{art2}. The scenario with $\lambda_1^2=\lambda_2^2=10^{-3}$ appears to be difficult for both methods as they do not achieve NMSE below $0$~dB. 
INV also fails when $\lambda_2^2=10^{-3}$ and $\lambda_1^2/\lambda_2^2$ 
corresponds to $-10$~dB (i.e., $\lambda_1^2=10^{-4}$) until $m\leq 17$. This is 
in accordance with the results of the previous example that shows that INV 
fails when $\lambda_1^2$, $\lambda_2^2$ and $d$ are ``too large''. The example 
here reveals one more detail: INV can benefit from smaller perturbations of the 
target component ($\lambda_1^2=10^{-4}$) even if $\lambda_2^2$ is larger, but 
the target dimension must be large enough with respect to $d$.

LS performs independently of $\lambda_2^2$, which is confirmed by the cases 
that are plotted with solid and dashed lines in Fig.~\ref{art2}: these lines 
coincide as both correspond to the same $\lambda_1^2$ (although different 
$\lambda_2^2$). LS is outperformed by INV when 
$\lambda_1^2=\lambda_2^2=10^{-4}$, which, again, occurs when the estimation 
error of the whole $\widehat\W$ is very small.

\subsection{Underdetermined model}\label{underexpsection}
\begin{figure}
\centering
\includegraphics[width=\linewidth]{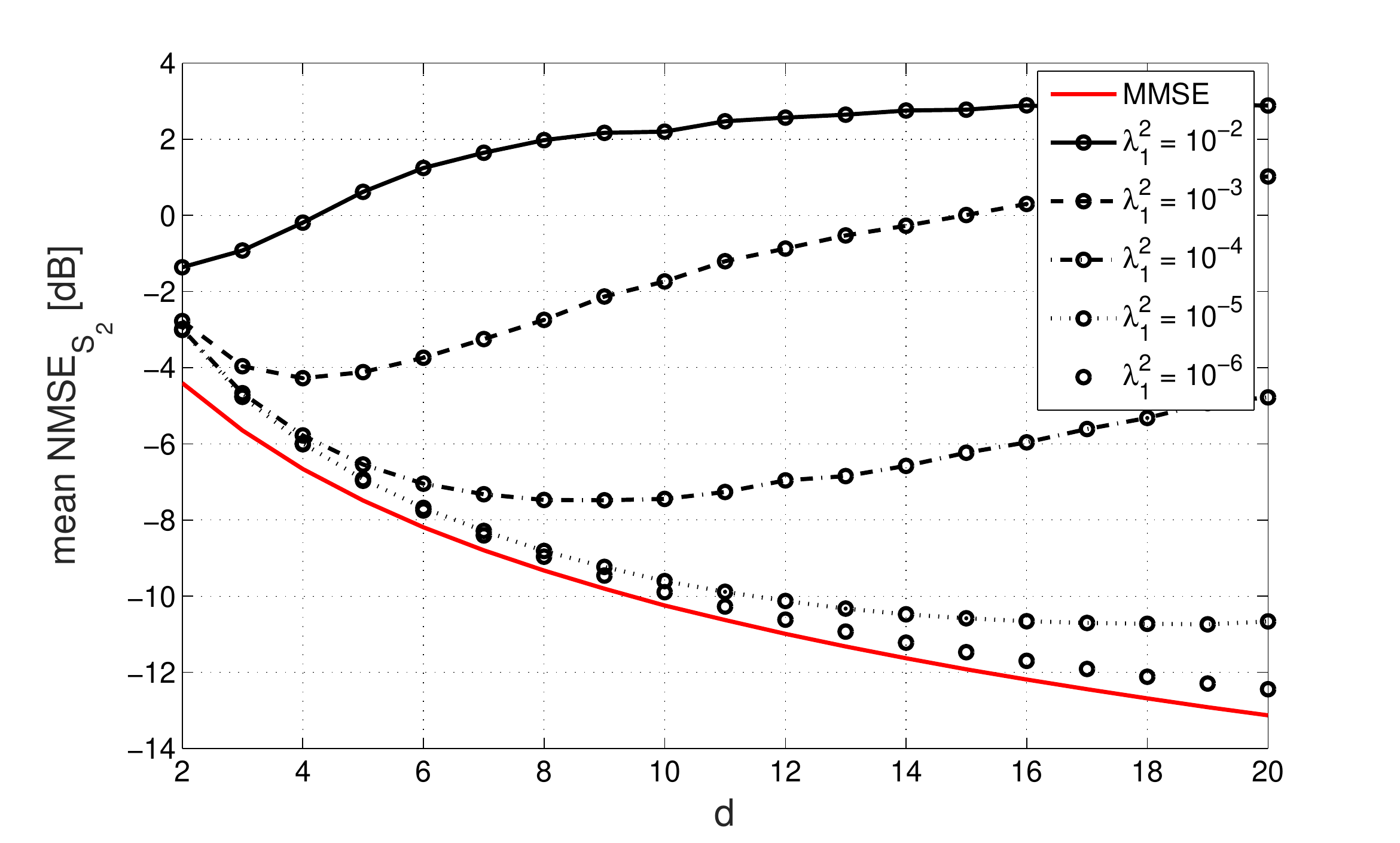}
\caption{\label{art4} Average NMSE$_{\bS_2}$ as a function of $d$ achieved by 
(\ref{est2}) in an experiment with the underdetermined 
model (\ref{umodel}); $m=1$; $N=10^4$. MMSE denotes the NMSE achieved by the 
optimum minimum mean-squared error solution (\ref{LSund}). The signals are 
generated as random complex Gaussian i.i.d.}
\end{figure}

\begin{figure}
\centering
\includegraphics[width=\linewidth]{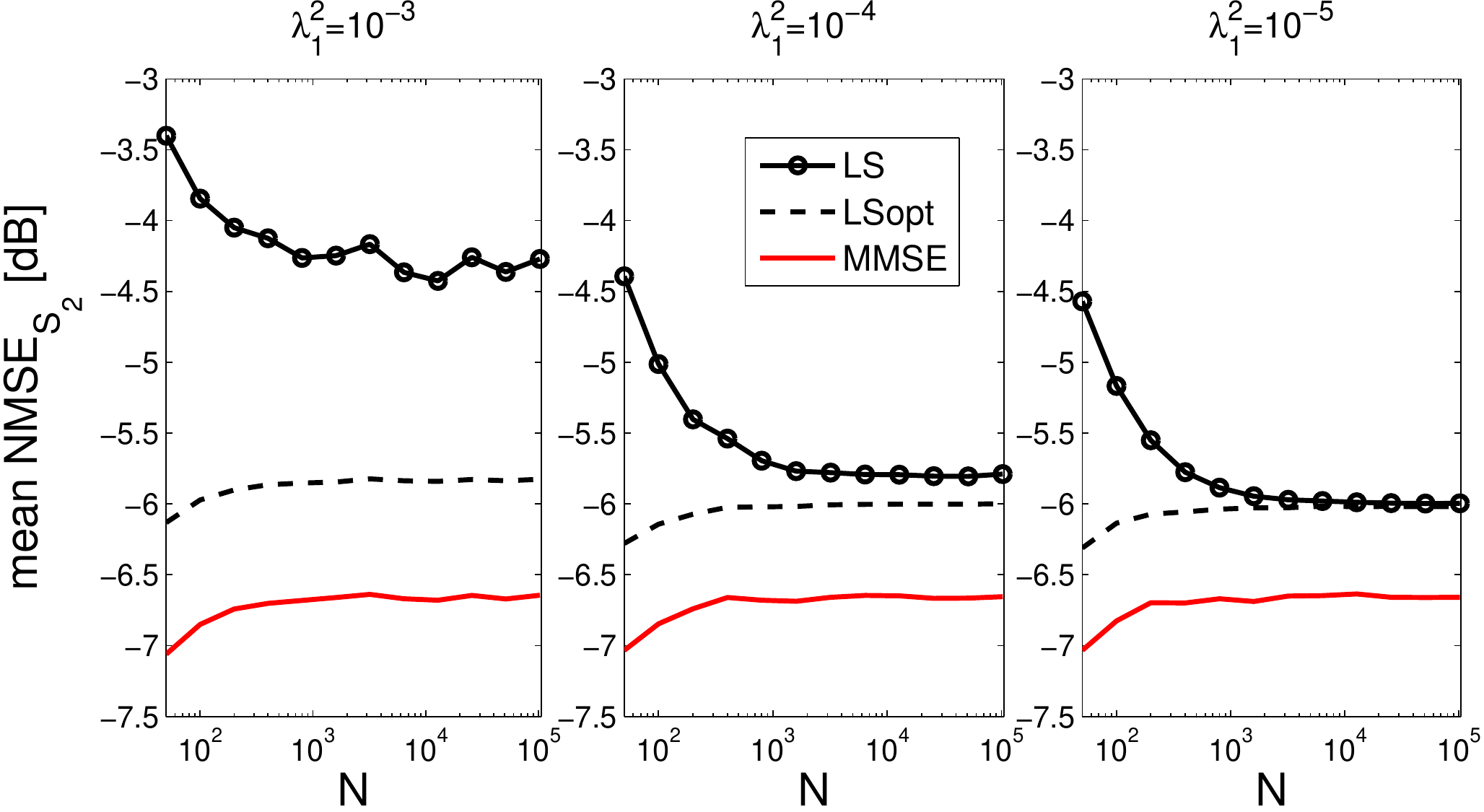}
\caption{\label{art4N} Average NMSE$_{\bS_2}$ as a function of $N$ achieved 
through	(\ref{est3}), \eqref{min2} and \eqref{LS}; $d=4$; $\lambda_1^2=10^k$, 
$k=-3,-4,-5$.}
\end{figure}

In the example of this subsection, we consider the underdetermined mixture 
model (\ref{umodel}) where $m=1$, $d=2,\dots,20$, and $N=50,\dots,10^5$. The 
goal is to 
examine the reconstruction of the noise components $\bS_2$ through 
(\ref{est3}). $\HH_1$ is randomly generated. Then, $\W$ is such that its rows 
form a basis of the $(d-m)$-dimensional subspace that is orthogonal to $\HH_1$ 
plus a random Gaussian perturbation matrix whose elements have the variance 
values equal to $\lambda_1^2=10^k$, $k=-2,-3,\dots,-6$, respectively. After 
applying (\ref{est3}), the evaluation is done using the normalized mean square 
distance 
	\begin{equation}\label{NMSEs}
		{\rm NMSE_{\bS_2}}=\frac{\|\bS_2-\widehat\bS_2\|_F^2}{\|\bS_2\|_F^2}.
	\end{equation}
%where $\widehat\s_2$ is the estimate of $\s_2$. 

Owing to the statement of Proposition~\ref{propMSE}, it is worth comparing 
$\widehat\bS_2$ with the exact solution of \eqref{min2}, which will be 
abbreviated by LSopt, and with the minimum mean square error solution, marked 
as MMSE, defined as the minimizer of
\begin{equation}\label{LSund}
	\min_{{\bf Q}\in\mathcal{C}_{d\times d}}\|\bS_2-{\bf Q}\X\|_F^2.
\end{equation}
The latter gives the minimum achievable value of NMSE$_{\s_2}$ by a linear 
estimator; cf. \eqref{min2} and \eqref{LS}.

The results averaged over $10^3$ independent trials are shown in 
Figures~\ref{art4} and \ref{art4N}. Fig.~\ref{art4} shows results for $N=10^4$. 
One observation here is that NMSE$_{\bS_2}$ achieved through LS is getting 
closer to that 
of MMSE as $\lambda_1^2$ approaches zero. Next, NMSE$_{\bS_2}$ improves with 
growing dimension $d$, but it appears that it stops improving at a certain $d$ 
and grows beyond this threshold value, which depends on $\lambda_1^2$. For 
example, when $\lambda_1^2=10^{-4}$, the NMSE$_{\bS_2}$ is decaying until $d=8$ 
and grows beyond $d\geq 10$. 

Fig.~\ref{art4N} shows NMSE$_{\bS_2}$ as a function of $N$ when $d=4$. This 
detailed observation shows that LS approaches LSopt as $N$ grows and 
$\lambda_1^2$ approaches zero, but does not achieve the performance of MMSE. 
This is the fundamental limitation due to the dimension of the separable signal 
subspace, that is, $d-m$.

\section{Practical Examples}

\subsection{De-noising of Electrocardiogram}
Fig.~\ref{ekgmix} shows two seconds of a recording from a three channel 
electrocardiogram (ECG) of a Holter monitor, which was sampled at $500$~Hz. The 
recording is strongly interfered with  a noise signal originating from the 
Holter display. The fundamental frequency of the noise is about $37$~Hz, and 
the noise contains several harmonics.

Since the noise is significantly stronger than the ECG components, Principal 
Component Analysis (PCA) can be used to find a demixing transform that 
separates the noise from the mixture. Therefore, we take the eigenvector 
corresponding to the highest eigenvalue of the covariance matrix of the 
recorded data (the principal vector) as the separating transform. Then, the 
noise responses on the electrodes are computed using LS and subtracted from the 
original noisy recording. This approach is computationally cheaper that doing 
the whole PCA and using INV then. According to Proposition~5, both 
approaches give the same result as PCA yields component that are exactly 
orthogonal.

To compare, we repeated the same experiment using the vector obtained through 
Independent Component Analysis (ICA). One-unit FastICA \cite{fastica} with 
$\tanh(\cdot)$ nonlinearity was used to compute the vector separating the noise 
component. To avoid the permutation ambiguity, the algorithm was initialized 
from $[1\,1\,1]$, because the noise appears to be uniformly distributed over 
the electrodes. Also here the approach is faster than doing the whole 
orthogonally-constrained ICA (e.g. using Symmetric FastICA) and using INV. %, 
%although the results would be slightly different (Proposition~5 does not apply 
%because one-unit FastICA does not apply the orthogonal constraint).

Figures \ref{ekgPCA} and \ref{ekgICA} show the resulting signals where the 
estimated images of the noise component were removed, respectively, through PCA 
and ICA. Both results show very efficient subtraction of the noise. A visual 
inspection of the detail in Fig.~\ref{ekgPCA} shows certain residual noise that 
does not appear in Fig.~\ref{ekgICA}, so the separation through ICA appears to 
be more accurate than by PCA. Combining one-unit ICA algorithm with LS, the 
computational complexity of the ICA solution is decreased.

\begin{figure}
	\centering
	\includegraphics[width=0.95\linewidth]{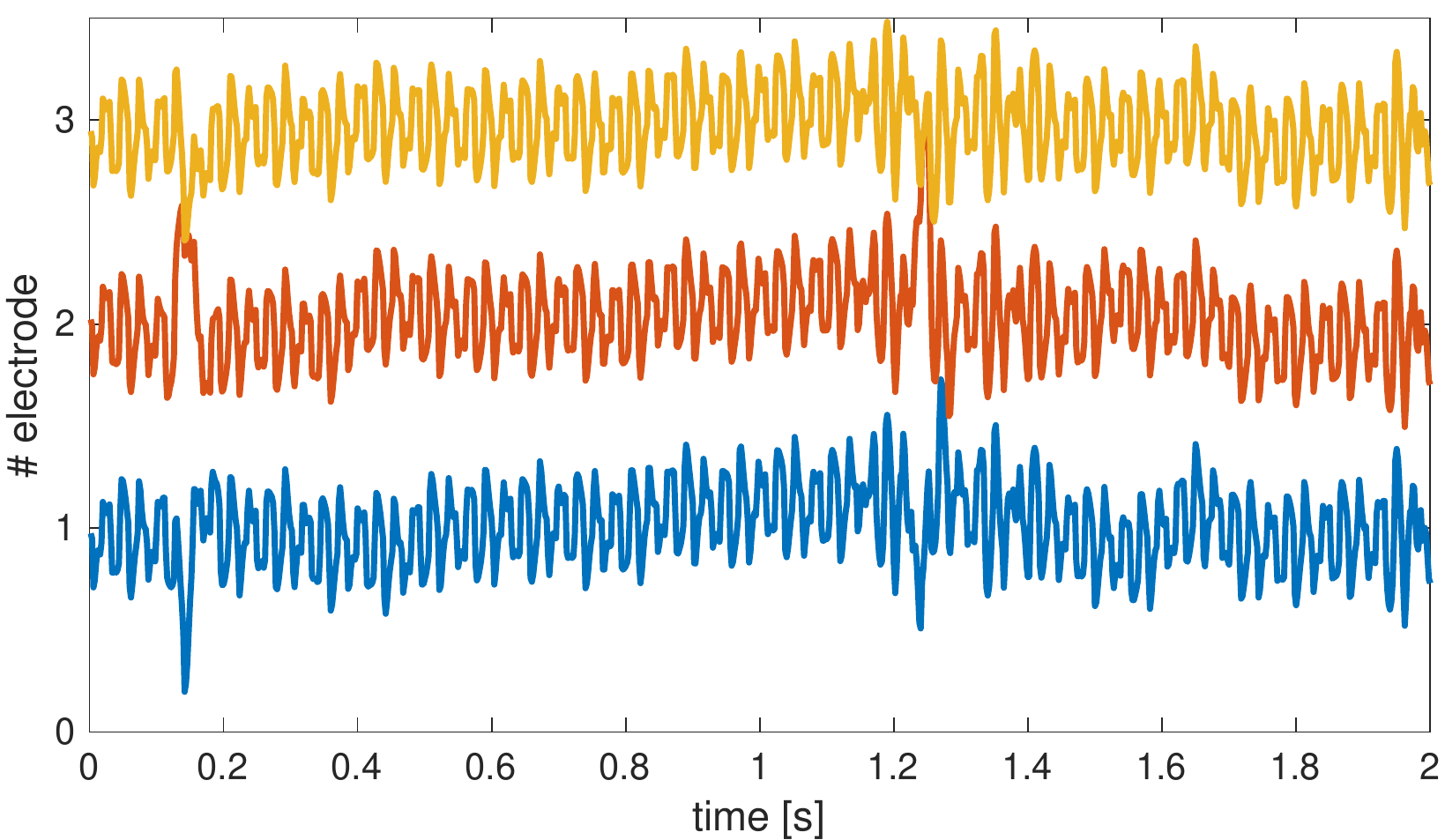}
	\caption{\label{ekgmix} A two-second sample of a three channel 
		electrocardiogram interfered by a noise signal originating from a 
		Holter 
		display.}
\end{figure}

\begin{figure}
	\centering
	\includegraphics[width=0.95\linewidth]{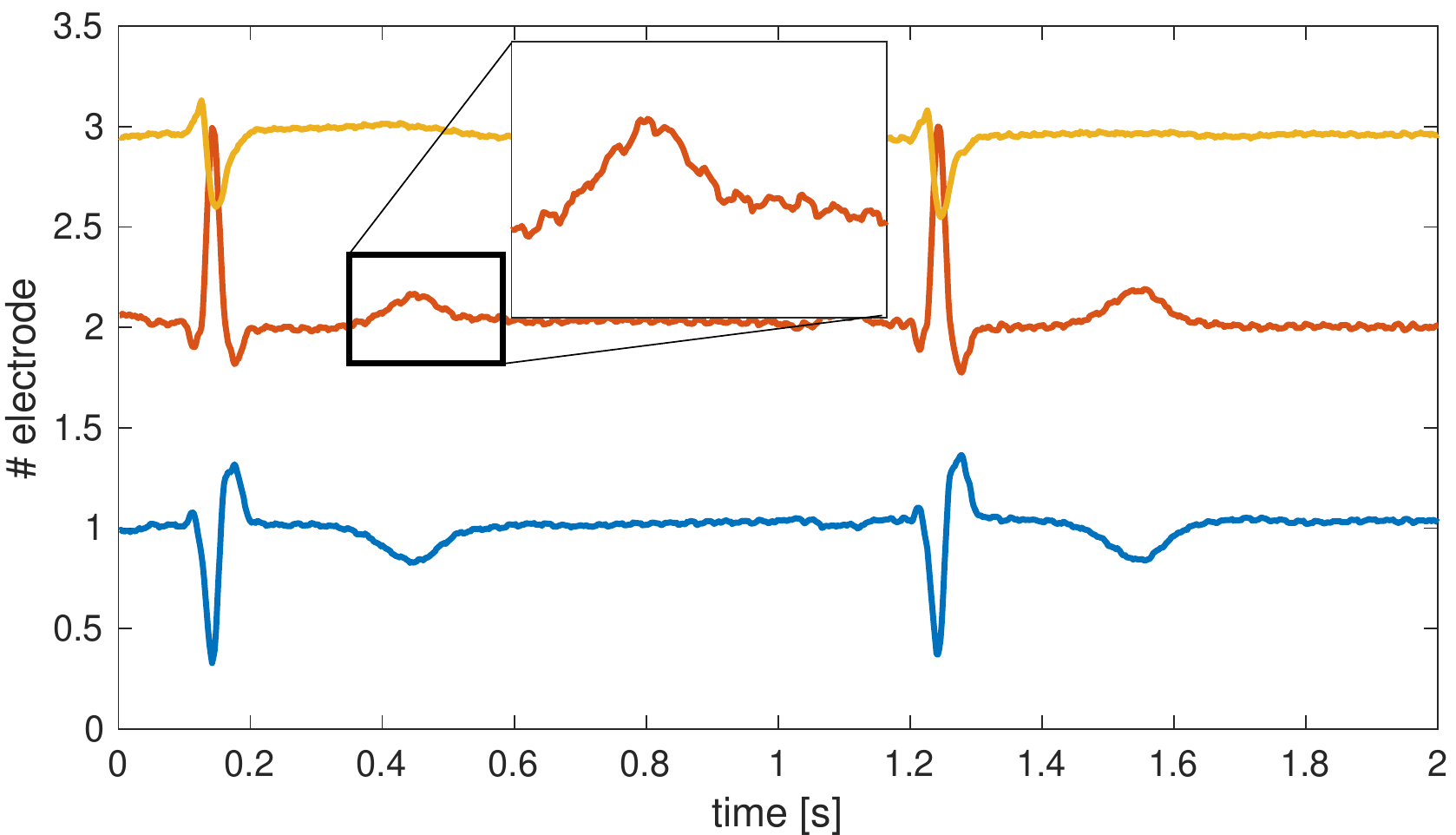}
	\caption{\label{ekgPCA} Cleaned data from Fig.~\ref{ekgmix} after the 
		subtraction of noise responses that were estimated through the main 
		principal 
		component and LS.}
\end{figure}

\begin{figure}
	\centering
	\includegraphics[width=0.95\linewidth]{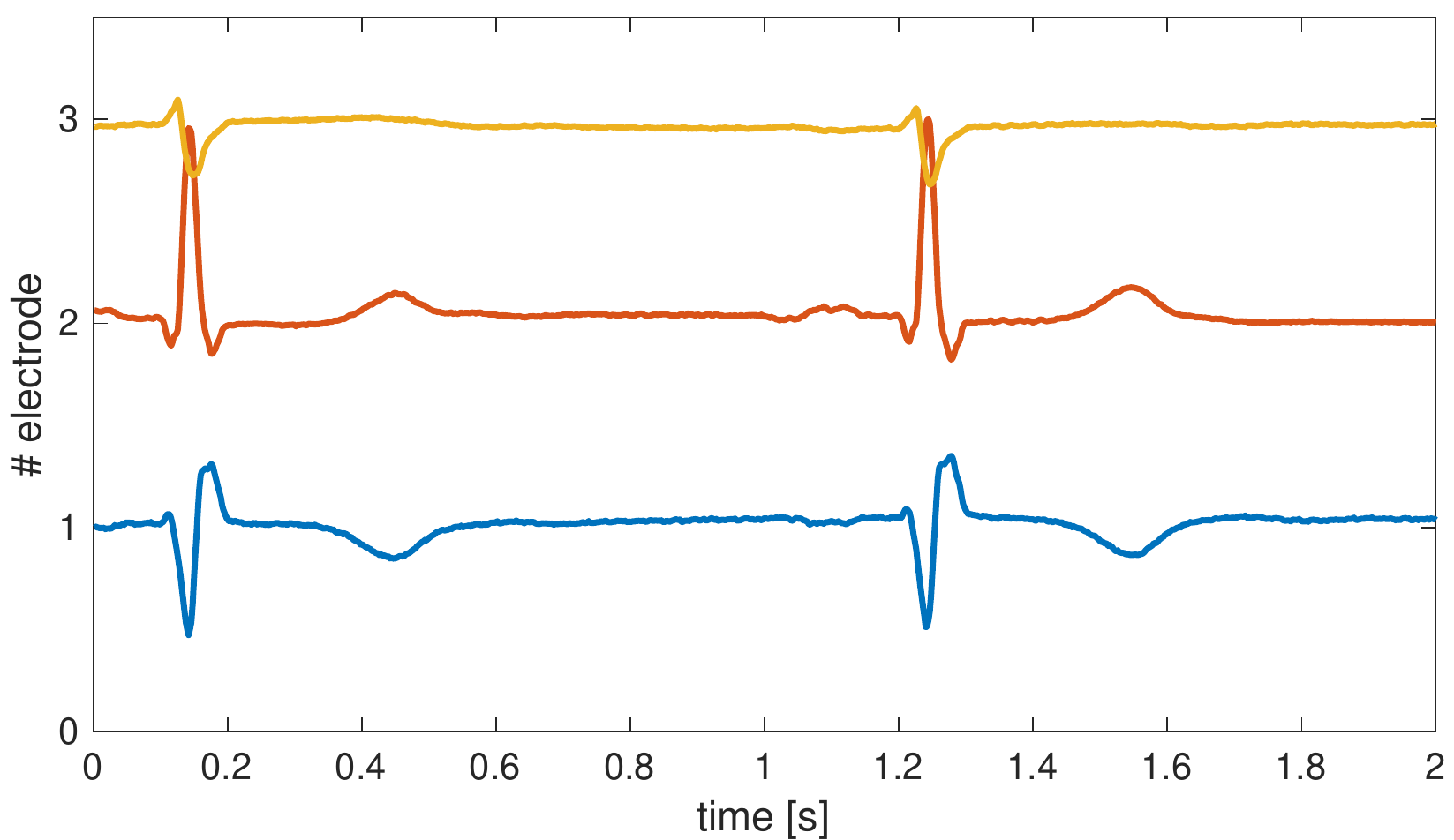}
	\caption{\label{ekgICA} Cleaned data from Fig.~\ref{ekgmix} after the 
		subtraction of noise images that were estimated using the one-unit 
		FastICA and 
		LS.}
\end{figure}

\subsection{Blind Separation with Incomplete Demixing 
Transforms}\label{IDTsection}
\begin{figure}
	\centering
	\includegraphics[width=\linewidth]{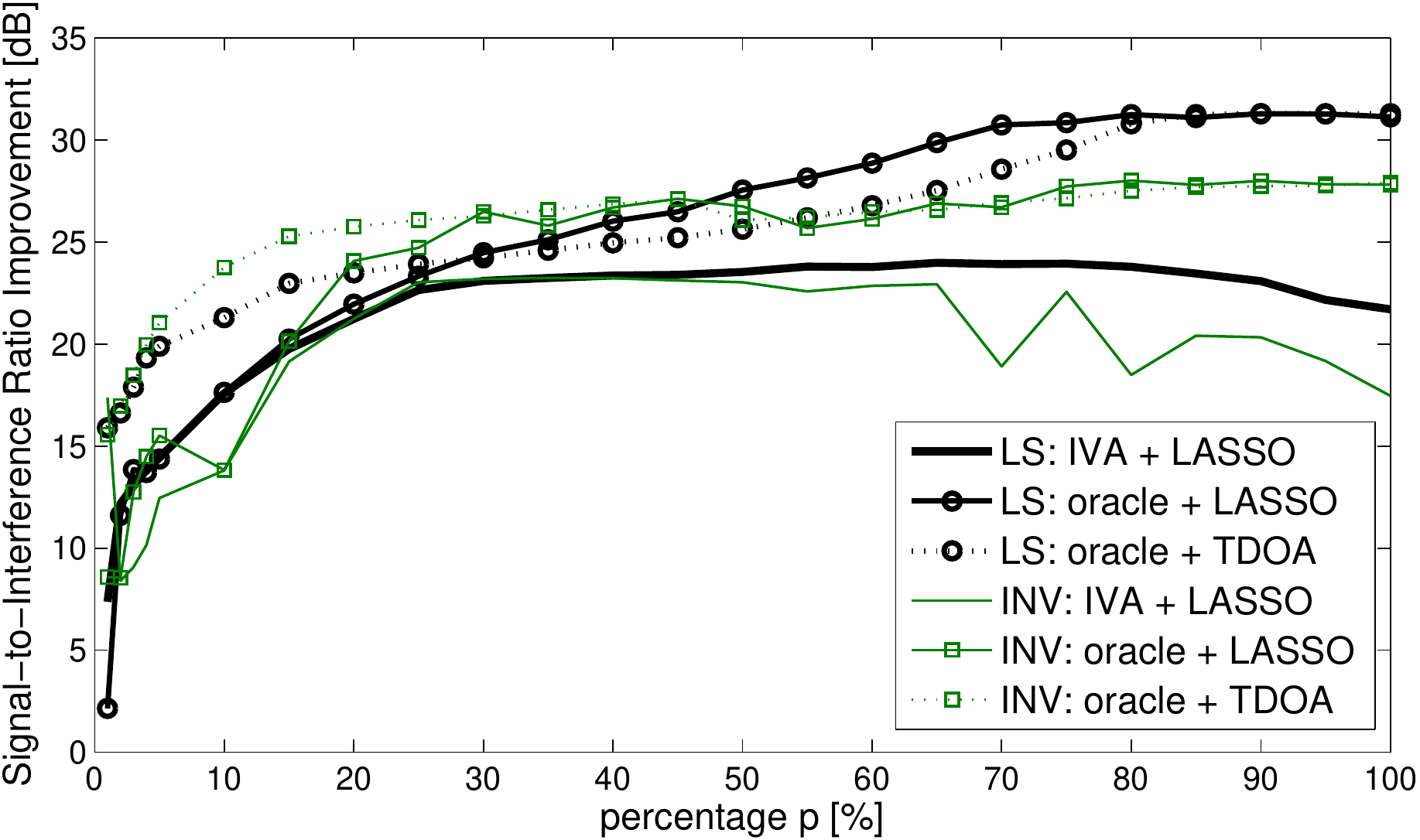}
	\caption{\label{IDTspeech} Improvement of signal-to-interference ratio as a 
		function of $p$, i.e., of percents of selected active frequencies in 
		$\mathcal{S}$ for the estimation of incomplete demixing transform. The 
		evaluation was performed on speech signals.}
\end{figure}
\begin{figure}
	\centering
	\includegraphics[width=\linewidth]{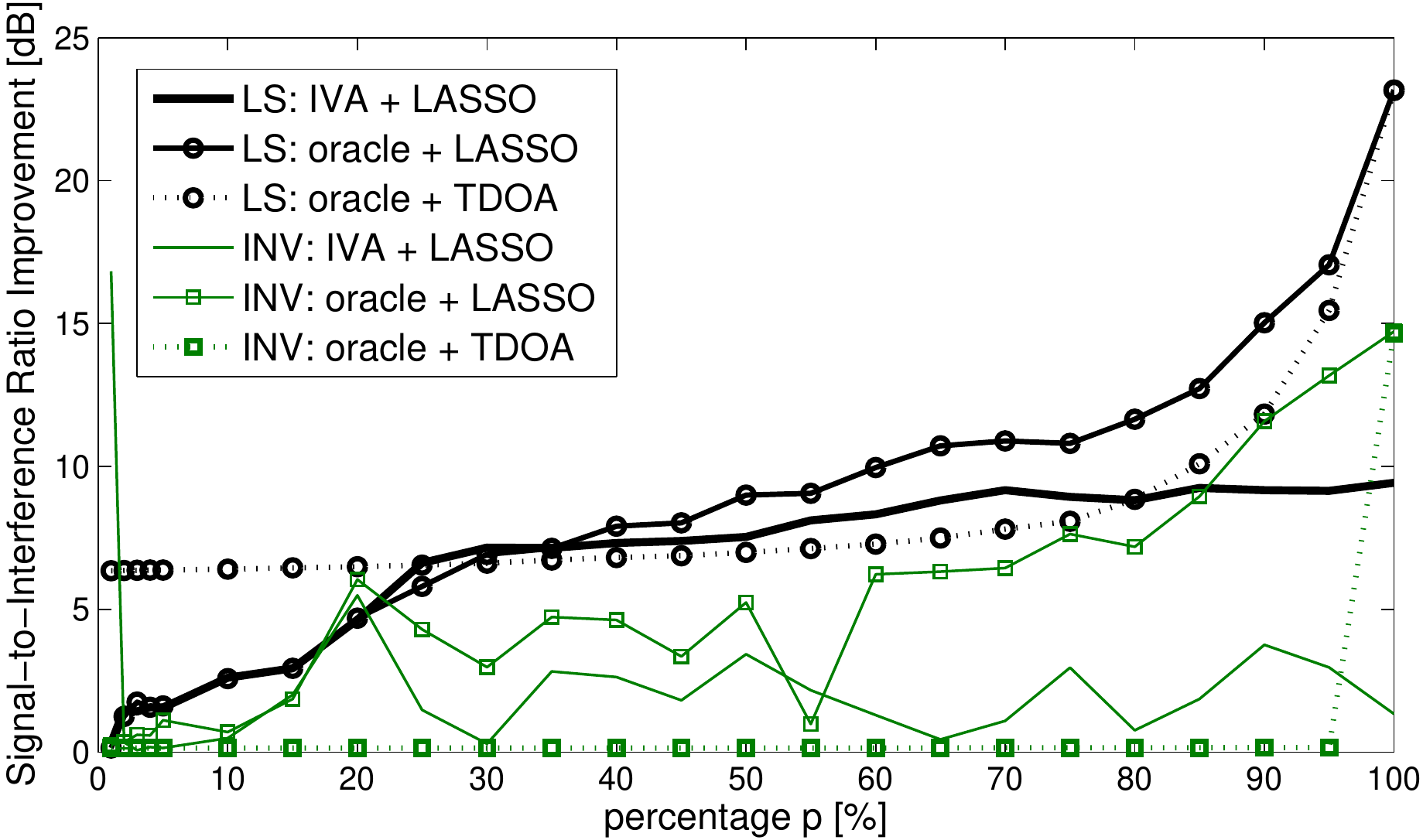}
	\caption{\label{IDT} Improvement of signal-to-interference ratio as a 
	function of $p$. The evaluation was performed on white noise signals that 
	uniformly excite the whole frequency range.}
\end{figure}

Proposition~5 points to the fact that if BSS is based on a method that yields 
(almost) orthogonal estimates of $\s_1$ and $\s_2$, INV and LS are principally 
not that different. The example of this section demonstrates a situation where 
the estimates are significantly nonorthogonal, so INV and LS yield considerably 
different results.

Recently, a novel approach for blind separation of convolutive mixtures of 
audio signals has been proposed in \cite{eusipco2016,iwaenc2016}. The idea 
resides in the application of IVA in the frequency-domain on a constrained 
subset of frequencies where the input (mixed) signals are active. This is in 
contrast with the conventional Frequency-Domain IVA (ICA), which is applied in 
all frequencies. The motivation behind is threefold: computational savings, 
improved accuracy (especially in environments with sparse room impulse 
responses), and the prediction of complete demixing transform for separation of 
future signals whose activity appears in other frequencies.

The method proceeds in three main steps. First, the subset $\mathcal{S}$ of $p$ 
percents of the most active frequency bins is selected. This can be done 
through estimating the power spectrum of the signal on a reference microphone  
using the coefficients of its short-term Fourier transform. The frequency bins 
with maximum average magnitude of Fourier coefficients are selected. Second, an 
IVA method is applied that estimates the demixing matrices within the subset of 
the selected frequencies. The subset of the matrices is referred to as 
Incomplete Demixing Transform (IDT). Third, the IDT is completed by a given 
method.

We consider the same experiment with two simultaneously 
speaking persons and two microphones as in \cite{eusipco2016}. Signals have 
10 seconds in length; the sampling frequency is 16~kHz. The 
signals are convolved with room impulse responses (RIR) generated by a 
simulator and mixed together. Reflection order is set to 1 so that the RIRs are 
significantly sparse (the results of this experiment with reflection order 10 
when RIRs are no more sparse are available in \cite{eusipco2016}). Then, the 
signals are transformed 
into the short-term Fourier domain with the window length of $1024$ samples and 
shift $128$. The convolution is, in the frequency domain, approximated by the 
set of multiplicative models \eqref{model} where $d=r=2$ where one model 
corresponds to one frequency bin; there are $513$ models in total. 

As for the second step, the demixing matrices are estimated 
from the mixed signals 
using the natural gradient algorithm for IVA \cite{iva} applied to the subset 
of models \eqref{model}.  To compare, 
``oracle'' demixing matrices are derived on $\mathcal{S}$ using known responses 
of the speakers. This gives the IDT that is known only on the selected subset 
$\mathcal{S}$.

The IDT is completed by two alternative methods. The first method, denoted as 
TDOA, utilizes known time-differences of arrival of the signals. The unknown 
demixing matrices are such that their rows correspond to the null beamformer 
steering 
spatial null towards the unwanted speaker. The second approach, denoted as 
LASSO, completes the IDT through finding the sparsest representations of 
incomplete relative transfer functions (RTF) that are derived from the 
IDT\footnote{As pointed in \cite{eusipco2016}, LASSO could be seen as a 
generalization of TDOA, because impulse responses corresponding to null 
beamformers are pure-delay filters, which are perfectly sparse.}. Let ${\bf q}$ 
denote an $|\mathcal{S}|\times 1$ vector that collects the 
coefficients of an incomplete RTF; $|\mathcal{S}|$ 
is the number of elements in $\mathcal{S}$. The completed RTF is obtained as 
the solution of \cite{tibshirani}
\begin{equation}\label{LASSO}
\arg\min_{{\bf h}} \|{\bf h}_{\mathcal{S}}-{\bf q}\|^2+  \epsilon\|{\bf 
F}^H{\bf h}\|_1,
\end{equation}
where $\epsilon>0$ controls the time-domain sparsity of the solution, ${\bf F}$ 
is the matrix of the DFT, the subscript $(\cdot)_{\mathcal{S}}$ denotes a 
vector/matrix with selected elements/rows whose indices are in $\mathcal{S}$, 
and $\|\cdot\|_1$ denotes the $\ell_1$-norm.

%To find the solution of (\ref{LASSO}), we use the fast proximal algorithm 
%proposed in \cite{CSRTF}, whose complexity is $\mathcal{O}(K\log{K})$ per 
%iteration\footnote{A Matlab implementation of the algorithm is available at 
%{\tt http://itakura.ite.tul.cz/zbynek/dwnld/SpaRIR.m}.}.

Now, it is worth noting that the separated components by the demixing matrices 
after the completion can be significantly nonorthogonal. While the IVA applied 
within $\mathcal{S}$ aims to find 
independent (thus ``almost'' or fully orthogonal) components, 
the method for the IDT completion does not take any regard to the 
orthogonality\footnote{The orthogonal constraint cannot be imposed within the 
frequencies outside of the set $\mathcal{S}$, because signals are not (or 
purely) active there.}.

Figures \ref{IDTspeech} and \ref{IDT} show results of the experiment from 
\cite{eusipco2016} evaluated in terms of the Signal-to-Interference Ratio 
Improvement (SIR) after the signals are separated as a function of $p$ (the 
percentage of frequencies in $\mathcal{S}$). 
In Fig.~\ref{IDTspeech}, the evaluation is performed with the speech signals, 
while Fig.~\ref{IDT} shows the results achieved as if the sources were white 
Gaussian sequences. The purpose of the latter evaluation is to 
evaluate the completed IDT uniformly over the whole frequency range, i.e., also 
in frequencies that were not excited by the speech signals. Note that SIR must 
be evaluated after resolving the scaling ambiguity in each frequency 
\cite{bsseval}. This gives us the opportunity to apply either INV or LS.

The results in Figures \ref{IDTspeech} and \ref{IDT} point to significant 
differences between LS and INV in this evaluation. The results by LS appear to 
be less biased and stable as compared to those by INV, and can be interpreted 
in accord with the theory. In particular, LS shows that oracle+LASSO (oracle 
IDT completed by LASSO) outperforms oracle+TDOA for $p$ between 35\% and 80\%. 
This gives sense, because LASSO can better exploit the sparsity of the RIRs 
generated in this experiment. The results by INV do not reveal this important 
fact. Next, LS shows in Fig.~\ref{IDTspeech} that IVA+LASSO can improve the 
separation of the speech signals when $p<100$\%. The evaluation on white noise 
in Fig.~\ref{IDT} shows that the loss of SIR is not essential until $p<30$\%. 
The latter conclusion cannot be drawn with the results by INV.

\section{Conclusions}
We have analyzed and compared two estimators of sensor images (responses) of 
sources that were separated from a multichannel mixture up to an unknown 
scaling factor: INV and LS. Simulations and perturbation analysis have shown  
pros and cons of the methods, which can be summarized into the following 
recommendations.
\begin{itemize}
	\item LS is more practical in a sense that the whole 
	mixing matrix need not be identified for its use, which is useful 
	especially in underdetermined scenarios.
	\item  The advantage of INV resides in the independence on 
	the (estimated) covariance matrix.
	\item INV could be beneficial as compared to LS when used with 
	non-orthogonal BSS algorithms, i.e., those not applying the orthogonal 
	constraint. However, both the target as well as the interference subspaces 
	must be estimated with a sufficient accuracy.	
\end{itemize}
Both approaches have been shown to be equivalent under the orthogonal 
constraint, so the differences in their accuracies are less significant when 
BSS yields signal components that are (almost) orthogonal (e.g. PCA, ICA, IVA). 
By contrast, the differences between the reconstructed images of nonorthogonal 
components can be large, as demonstrated in the example of 
Section~\ref{IDTsection}.

\section*{Appendix: Asymptotic Expansions}
\subsection*{Computation of (\ref{MDPan1})}
Let ${\bf E}$ contain first $m$ columns of the $d\times d$ identity matrix. It follows that
\begin{align*}
\HH{\bf E}&=\HH_1, & {\bf A}{\bf E}&={\bf A}_1, \\
{\bf E}^H\W&=\W_1, & {\bf E}^H{\bf V}&={\bf V}_1.
\end{align*}
To derive an approximate expression for ${\bf A}$, we will use the first-order expansion
\begin{align}
{\bf A}&={\bf V}^{-1}=(\W+\bXi)^{-1}=({\bf I}+\HH\bXi)^{-1}\HH\\
&\approx ({\bf I}-\HH\bXi)\HH=\HH-\HH\bXi\HH.
\end{align}
Now we apply this approximation and neglect terms of higher than the first order.
\begin{multline}
\left\|\HH_1\W_1-\A_1{\bf V}_1\right\|_F^2=
\left\|\HH_1\W_1-\A{\bf E}{\bf E}^H{\bf V}\right\|_F^2\approx\\
\left\|\HH_1\W_1-(\HH-\HH\bXi\HH){\bf E}{\bf E}^H(\W+\bXi)\right\|_F^2=\\
\left\|\HH_1\W_1-(\HH_1-\HH\bXi\HH_1)(\W_1+\bXi_1)\right\|_F^2\approx\\
\left\|\HH\bXi\HH_1\W_1-\HH_1\bXi_1\right\|_F^2.
\end{multline}
\hfill\rule{1.2ex}{1.2ex}

\subsection*{Computation of (\ref{proposedan1})}
We start with the first approximation
\begin{multline}
\left\|\HH_1\W_1-\widehat{\bf C}{\bf V}_1^H({\bf V}_1\widehat\C{\bf 
V}_1^H)^{-1}{\bf V}_1\right\|_F^2=\\
\Bigl\|\HH_1\W_1-({\bf C}+\Delta\C)(\W_1^H+\bXi_1^H) \\ 
\cdot((\W_1+\bXi_1)({\bf C}+\Delta\C)(\W_1^H+\bXi_1^H))^{-1} 
(\W_1+\bXi_1)\Bigr\|_F^2\approx\\
\Bigl\|\HH_1\W_1-({\bf C}+\Delta\C)(\W_1^H+\bXi_1^H)\cdot(\W_1\C\W_1^H+ \\ 
\W_1\Delta\C\W1^H+\bXi_1\C\W_1^H+\W_1\C\bXi_1^H)^{-1}(\W_1+\bXi_1)\Bigr\|_F^2.
\end{multline}
Since $\W$ is now the exact inverse of $\HH$, it holds that 
$\W_1\C\W_1^H=\C_{\s_1}$. By neglecting higher than the first-order terms and 
by applying the 
first-order expansion of the matrix inverse inside the expression,
\begin{multline}\label{exp3}
\Bigl\|\HH_1\W_1-({\bf C}+\Delta\C)(\W_1^H+\bXi_1^H)({\bf 
I}+\C_{\s_1}^{-1}\bXi_1\C\W_1^H+ \\ 
+\C_{\s_1}^{-1}\W_1\Delta\C\W_1^H+\C_{\s_1}^{-1}\W_1\C\bXi_1^H)^{-1}\C_{\s_1}^{-1}(\W_1+\bXi_1)\Bigr\|_F^2\approx\\
\Bigl\|\HH_1\W_1-(\C\W_1^H+\C\bXi_1^H+\Delta\C\W_1^H)({\bf 
I}-\C_{\s_1}^{-1}\bXi_1\C\W_1^H- \\ 
-\C_{\s_1}^{-1}\W_1\Delta\C\W_1^H-\C_{\s_1}^{-1}\W_1\C\bXi_1^H)\C_{\s_1}^{-1}(\W_1+\bXi_1)\Bigr\|_F^2.
\end{multline}
Since, 
\begin{align}
\C\W_1^H\C_{\s_1}^{-1}&=\HH\,{\tt 
bdiag}(\C_{\s_1},\C_{\s_2})\HH^H\W_1^H\C_{\s_1}^{-1}\nonumber\\ 
&=\HH_1,
\end{align} 
the zero order term in (\ref{exp3}) vanishes. By neglecting higher than the first-order terms, (\ref{proposedan1}) follows.
\hfill\rule{1.2ex}{1.2ex}

\section*{Acknowledgments}
This work was supported by The Czech Science Foundation through Project No.~14-11898S and partly by California Community Foundation through Project No.~DA-15-114599.

We thank BTL Medical Technologies CZ for providing us the three-channel ECG recording.

\begin{biography}[{\includegraphics[width=1in,height=1.25in,clip,keepaspectratio]{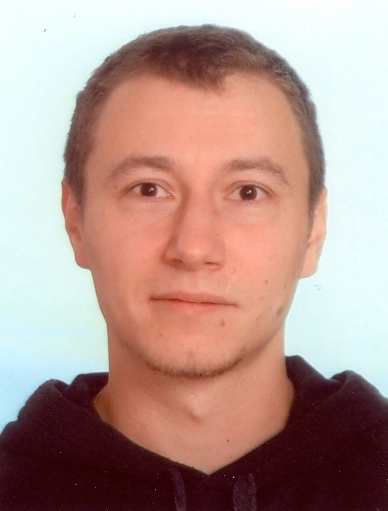}}]{Zbyn\v{e}k
		Koldovsk\'{y}}
	(M'04-SM'15) was born in Jablonec nad Nisou, Czech Republic, in 1979. He 
	received the M.S. degree and Ph.D. degree in mathematical modeling from 
	Faculty of Nuclear Sciences and Physical Engineering at the Czech Technical
	University in Prague in 2002 and 2006, respectively. He was also with the 
	Institute of Information Theory and Automation of the Academy of Sciences 
	of the Czech Republic from 2002 to 2016.
	
	Currently, he is an associate professor at the Institute of Information 
	Technology and Electronics, Technical University of Liberec, and the leader 
	of Acoustic Signal Analysis and Processing (A.S.A.P.) Group. He is the 
	Vice-dean for Science, Research and Doctoral Studies at the Faculty of 
	Mechatronics, Informatics and Interdisciplinary Studies. His main research 
	interests are focused on audio signal processing, blind source separation, 
	independent component analysis, and sparse representations.
	
    Zbyn\v{e}k Koldovsk\'{y} has served as a general
	co-chair of the 12th Conference on Latent Variable Analysis and Signal 
	Separation (LVA/ICA 2015) in Liberec, Czech Republic.
	
\end{biography}

\begin{biography}[
	{\includegraphics[width=1in,height=1.25in,clip,keepaspectratio]{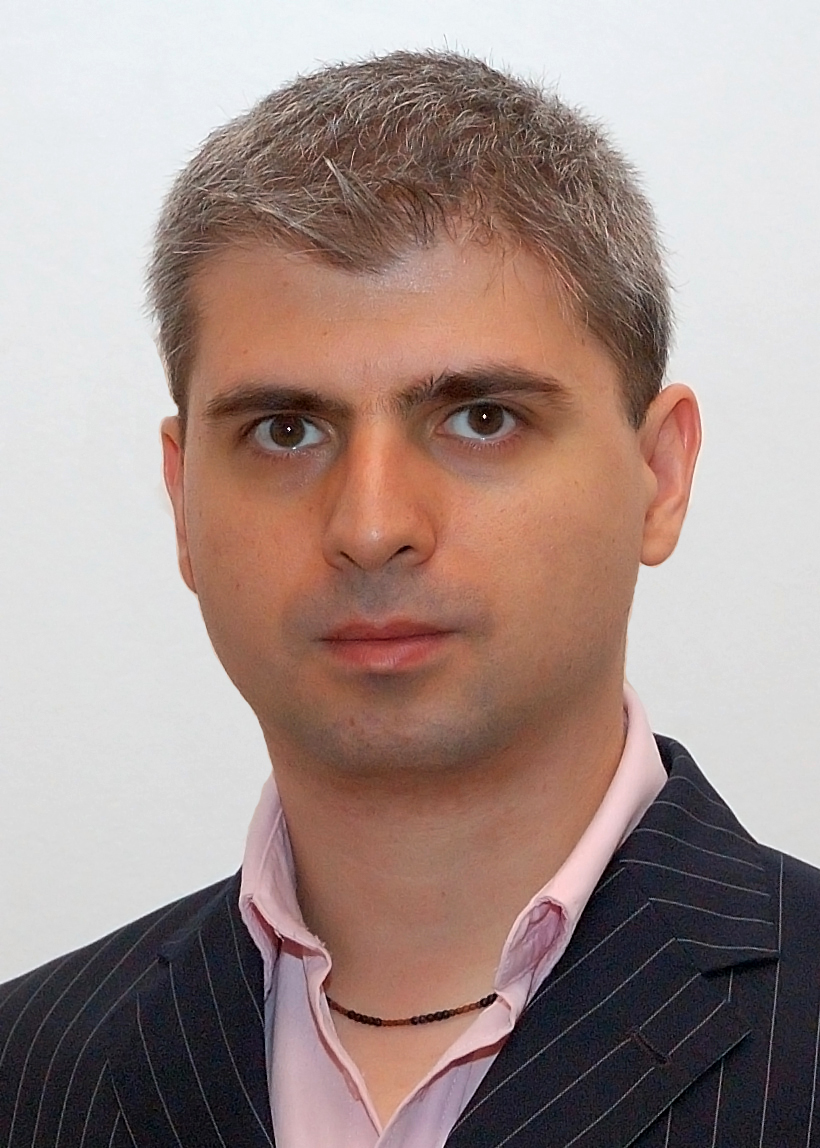}}]
	{Francesco Nesta} received the Laurea degree in computer engineering from
	Politecnico di Bari, Bari, Italy, in September 2005 and the Ph.D. degree
	in information and communication technology from University of Trento, 
	Trento,
	Italy, in April 2010, with research on blind source separation and
	localization in adverse environments.
	
	He has been conducting his research at Bruno Kessler Foundation IRST, Povo 
	di
	Trento, from 2006 to 2012. He was a Visiting Researcher from September 2008 
	to April 2009 with the Center for Signal and Image Processing Department, 
	Georgia Institute of Technology, Atlanta. His major interests include 
	statistical signal processing, blind source separation, speech enhancement, 
	adaptive filtering, acoustic echo cancellation, semi-blind source 
	separation and multiple acoustic source localization. He is currently 
	working at Conexant System, Irvine (CA, USA) on the development of audio 
	enhancement algorithms for far-field applications.
	
	Dr. Nesta serves as a reviewer for several journals such as the {\scshape 
		IEEE
		Transaction on Audio, Speech, and Language Processing}, {\em Elsevier 
		Signal Processing Journal}, {\em Elsevier Computer Speech and 
		Language}, and in several conferences and workshops in the field of 
	acoustic signal processing. He has served as Co-Chair in the third 
	community-based Signal Separation Evaluation Campaign (SiSEC 2011) and 
	as organizer of the 2nd CHIME challenge.
\end{biography}

\end{document}